\documentclass[preprint,1p,times]{elsarticle}
\usepackage{geometry}

\usepackage{graphicx} % Required for inserting images
\usepackage{chemformula} % Formula subscripts using \ch{}
\usepackage[T1]{fontenc} % Use modern font encodings

\usepackage{natbib}
\usepackage{amsmath,amsthm,amssymb}
\usepackage{float}
\usepackage{graphicx}% Include figure files
\usepackage{caption}
\usepackage{subfigure}
\usepackage{tabularx}% Include figure files
\usepackage{threeparttable}
\usepackage{dcolumn}% Align table columns on decimal point
\usepackage{bm}% bold math
\usepackage{color,epsfig,multirow}
\usepackage{array}
\usepackage{hyperref}
\usepackage{algpseudocode}%[noend]
\usepackage{algorithmicx,enumerate,algorithm}
\usepackage{booktabs}
\usepackage{threeparttable}
\usepackage{makecell}

\DeclareMathOperator{\Kn}{Kn}

\newcommand{\cO}{\mathcal{O}}

\newtheorem{theorem}{Theorem}

\theoremstyle{definition}

\theoremstyle{remark}
\newtheorem{remark}[theorem]{Remark}
\newtheorem{proposition}{Proposition}

\allowdisplaybreaks[4]
\begin{document}

\begin{frontmatter}

\title{From molecular dynamics to kinetic models: data-driven generalized collision operators in 1D3V plasmas} %% Article title
% Generalized numerical simulation of 1D3V plasma systems with data-driven approach
\author[1]{Yue Zhao}
\ead{zhaoyu14@msu.edu}
\address[1]{{Department of Computational Mathematics, Science \& Engineering, Michigan State University}, 428 S Shaw Ln, East Lansing 48824, MI, USA}
\author[3]{Guosheng Fu}
\ead{gfu@nd.edu}
\address[3]{{Department of Applied and Computational Mathematics and Statistics (ACMS), University of Notre Dame}, Notre Dame 46556, IN, USA}
\author[1,2]{Huan Lei}
\ead{leihuan@msu.edu}
\address[2]{{Department of Statistics \& Probability, Michigan State University}, 619 Red Cedar Road Wells Hall, East Lansing 48824, MI, USA}

% \date{\today}

%\affiliation{organization={},%Department and Organization
%	addressline={}, 
%	city={},
%	postcode={}, 
%	state={},
%	country={}}

%% Abstract
\begin{abstract}
We present a data-driven approach for constructing generalized collisional kinetic models for inhomogeneous plasmas in one-dimensional physical space and three-dimensional velocity space (1D–3V). The collision operator is directly learned from micro-scale molecular dynamics (MD)
and accurately accounts for the unresolved particle interactions over a broad range of plasma conditions. Unlike the standard Landau operator, the present operator takes an anisotropic, non-stationary form that captures the heterogeneous collisional energy transfer arising from the many-body interactions, which is crucial for plasma kinetics beyond the weakly coupled regime.  Efficient numerical evaluation is achieved through a low-rank tensor representation with $\cO(N \log N)$ computational complexity. The constructed kinetic equation strictly preserves conservation laws and physical constraints and therefore, enables us to develop an explicit second-order, energy-conserving scheme that ensures fully discrete conservation of mass and total energy. Numerical results demonstrate that the present model accurately predicts both transport coefficients and several 1D-3V kinetic processes compared with MD simulations across a broad range of densities and temperatures in spatially inhomogeneous settings.  This work provides a systematic pathway for bridging micro-scale MD and inhomogeneous plasma kinetic descriptions where empirical models show limitation.

% Building on molecular dynamics (MD) simulations, we extend a previously developed fast spectral separation approach to incorporate explicit dependence of the collision operator on local density and temperature, enabling accurate modeling beyond the weakly coupled regime.
% The resulting collision operator features an anisotropic, non-stationary kernel that adapts to evolving macroscopic states. 
% Efficient numerical evaluation is achieved through a low-rank tensor product representation combined with a Fourier spectral discretization, leading to an $\cO(N \log N)$ algorithm suitable for high-dimensional kinetic simulations. 

% The advection term is treated using explicit second-order, energy-conserving schemes to ensure fully discrete conservation of mass and total energy, while maintaining long-term stable simulations.
% % over a long period of time.
% % We adopt explicit second-order energy-conserving numerical schemes for the advection term, preserving fully discrete mass and total energy conservation.
% % The proposed scheme preserves fundamental physical properties, including conservation laws and a discrete H-theorem.
% Numerical experiments demonstrate that the proposed model accurately reproduces collisional relaxation dynamics observed in MD simulations across a range of densities and temperatures in spatially inhomogeneous settings, while maintaining high computational efficiency.
% This work provides a systematic pathway for bridging MD and kinetic modeling in inhomogeneous plasma systems.
\end{abstract}



%% Keywords

\end{frontmatter}

% \maketitle

\section{Introduction}

Collisional kinetic theory provides a fundamental framework for modeling the evolution of particle distribution functions in plasmas and rarefied gases, bridging microscopic particle interactions and macroscopic transport phenomena. 
In spatially inhomogeneous settings, the kinetic description typically consists of a Vlasov-type transport term coupled with a collision operator that accounts for unresolved particle interactions beyond the mean-field approximation. 
Among various models, the Landau collision operator \cite{landau1937kinetic} has been widely adopted for plasmas with long-range Coulomb interactions, owing to its rigorous conservation properties and consistency with the H-theorem. 
It can be formally derived from the Boltzmann collision operator \cite{boltzmann1872weitere, wild1951boltzmann} in the grazing collision limit or from the Balescu–Lenard equation \cite{lenard1960bogoliubov, balescu1960irreversible} under appropriate dielectric screening assumptions.

In practical plasma simulations, spatial inhomogeneity plays a crucial role in determining transport processes, wave–particle interactions, and relaxation dynamics. 
As a result, kinetic equations posed in one-dimensional physical space and three-dimensional velocity space (1D–3V) have become a standard modeling paradigm in plasma physics \cite{cheng1976integration}. 
Such models retain the essential velocity-space physics while allowing for nontrivial spatial transport, and have been extensively used in applications ranging from fusion plasmas to space and astrophysical systems. 
However, incorporating collisions in 1D–3V kinetic models significantly increases both analytical and computational complexity, as the collision operator must be evaluated locally in physical space while preserving its intrinsic high-dimensional structure in velocity space.
% resulting in substantial computational cost and complicated numerical coupling between transport and collisional dynamics

The numerical solution of collisional kinetic equations in 1D–3V phase space presents substantial challenges. Even for the classical Landau operator, the collision term involves a nonlinear, three-dimensional integro-differential operator evaluated at every spatial grid point, leading to extremely high computational cost. 
Considerable efforts have been devoted to developing efficient numerical methods for collisional kinetic equations in spatially inhomogeneous settings. 
Early approaches include finite-difference and finite-volume discretizations combined with conservative formulations to preserve mass, momentum, and energy \cite{chang1970practical,larsen1985discretization}. 
Entropy-based and structure-preserving schemes have been further developed to ensure discrete consistency with the H-theorem \cite{buet1999numerical,degond1994entropy}.

To alleviate the computational burden, various fast algorithms have been proposed for evaluating the Landau collision operator, including Fourier spectral methods that exploit the convolution structure of the isotropic Landau kernel \cite{pareschi2000fast,filbet2002numerical,zhang2017conservative}, multigrid and multipole-based approaches \cite{buet1997fast,lemou1998multipole},  Hermite spectral based methods \cite{wang2019approximation,li2020approximation,li2021hermite}, as well as formulations based on Rosenbluth potentials \cite{Rosenbluth_PR_1957} coupled with Poisson solvers \cite{chacon2000implicit,taitano2016adaptive}. Particle-based and hybrid methods have also been explored to improve scalability and robustness in high-dimensional settings \cite{takizuka1977binary,Carrillo_JCP_2020,bailo2024collisional,Hauck_Hu_JSP_2025}. 
To address the stiffness induced by disparate temporal scales between transport and collisions, implicit and asymptotic-preserving schemes have been proposed \cite{epperlein1994implicit,chacon2000implicit,lemou2005implicit,filbet2010class,jin2011class}. 
Despite these advances, the evaluation of the collision operator remains a dominant computational bottleneck in large-scale 1D–3V kinetic simulations.

Beyond computational considerations, a more fundamental limitation arises from the modeling assumptions underlying the Landau collision operator. 
The Landau formulation is valid in the weakly coupled regime, where binary interactions are dominated by small-angle scattering and particle correlations are negligible.  However, in many physically relevant scenarios, such as dense laboratory plasmas, warm dense matter, and inertial confinement fusion, particle correlations and collective effects become increasingly important. 
In these regimes, the Landau operator may fail to accurately capture collisional relaxation and energy transfer processes, even if it is evaluated with high numerical fidelity. This modeling deficiency poses a significant challenge for spatially inhomogeneous kinetic simulations, where local thermodynamic conditions can vary strongly in space and time.

To overcome these limitations, data-driven approaches have recently emerged as a promising pathway for constructing effective collision models informed by microscopic particle dynamics. 
Molecular dynamics (MD) simulations provide a first-principles description of interparticle interactions beyond weak-coupling assumptions, offering a natural foundation for learning generalized collisional behavior \cite{Stanton_Murillo_PRE_2016,baalrud2013effective}. 
In our previous work \cite{zhao2025data}, we introduced a data-driven collision operator for one-component plasmas that generalizes the classical Landau formulation beyond the weakly coupled regime. 
The operator is characterized by a non-stationary and symmetry-breaking collision kernel that captures correlation-induced anisotropic collisional effects absent from the classical Landau formulation. 
This generalized collision kernel was shown to accurately model moderately coupled plasma dynamics in spatially homogeneous (0D–3V) settings.
% preserves frame indifference, conservation laws, and entropy dissipation
In a subsequent study \cite{zhao2025fast}, we further developed a fast spectral separation representation of the generalized collision kernel, which exploits a low-rank tensor decomposition to recover a convolution structure and enables efficient evaluation with $\mathcal{O}(N\log N)$ computational complexity.
However, both the generalized collision model and its spectral separation formulation were restricted to spatially homogeneous 0D–3V systems, and show limitations to model spatially inhomogeneous collisional kinetics.

% collision kernel that depends not only on the relative velocity of colliding particles, but also on their average velocity
% adopts a metriplectic structure \cite{morrison1986paradigm}
% The anisotropic and non-stationary nature of the learned collision kernel breaks the convolution structure exploited by many fast algorithms for the Landau operator, rendering traditional Fourier-based or Rosenbluth-potential-based accelerations inapplicable.

In this work, we aim to extend the data-driven collision framework to spatially inhomogeneous 1D–3V kinetic systems within a general metriplectic formulation \cite{morrison1986paradigm,grmela1997dynamics}.
In contrast to the homogeneous case, the 1D–3V Vlasov–collision system involves nontrivial coupling between the spatial transport under a self-consistent electric field and the collisional process, which requires a careful balance between Hamiltonian transport and dissipative dynamics at both continuum and discrete levels.
From the metriplectic perspective, this extension introduces additional structural challenges: the collision operator must be consistently coupled with the Vlasov dynamics while preserving conservation laws, entropy dissipation, and compatibility with spatial inhomogeneity. 
In particular, the collision process further depends on the local plasma conditions that introduce spatially heterogeneous dissipative mechanisms. 
To capture this effect, we construct a generalized collisional kernel directly learned from the micro-scale MD simulations that explicitly captures the dependence on the local macroscopic moments of the distribution function. 
The developed model provides a consistent 1D–3V kinetic description of collisional plasma dynamics over a broad range of physical regimes. 
Moreover, the constructed model strictly preserves physical constraints and frame-indifference symmetry conditions, and therefore, enables the development of a structure-preserving discretization of the collisional kinetic equation.  
Numerical experiments demonstrate the accuracy of the constructed model for plasma in moderately coupled regimes (i.e., Coulomb coupling parameter $\Gamma \sim \mathcal{O}(1)$) through a direct comparison with the full MD simulations, where the Landau model shows limitations. 
Furthermore, the fully discrete numerical scheme can strictly guarantee the conservation of mass and energy.
% The developed model provides a structure-preserving 1D–3V kinetic description that couples the Vlasov equation with a generalized collision operator in a unified metriplectic form, providing a systematic foundation for modeling collisional plasma dynamics beyond the weakly coupled regime in spatially inhomogeneous configurations.
% Numerical experiments demonstrate that the present model significantly improves upon the classical Landau equation in reproducing MD-observed collisional dynamics in spatially inhomogeneous settings, while maintaining computational efficiency compatible with large-scale kinetic simulations.
% yields improved agreement with MD results compared to the classical Landau model.

The rest of this paper is organized as follows. 
Section \ref{sec:method} introduces the generalized collision operator with macroscopic state dependence, the spectral separation representation, and the fast evaluation strategy for 1D–3V systems. 
Section \ref{sec:numerical} presents numerical results illustrating the accuracy, efficiency, and structure-preserving properties of the proposed method. Summary and concluding remarks are given in Section~\ref{sec:summary}.

\section{Methods}\label{sec:method}

We consider the Vlasov-Amp\`ere-collision (VAC) model of the evolution of the ion probability distribution function with uniform background electrons, involving the advection and collision parts with zero magnetic field as follows,
\begin{equation}
    \dfrac{\partial f}{\partial t}(\bm{x},\bm{v},t) + \bm{v} \cdot \nabla_{\bm{x}} f + \dfrac{q_{0}}{m_{0}} \bm{E}(\bm{x}) \cdot \nabla_{\bm{v}} f  = C[f] , \quad \dfrac{\partial \bm{E}}{\partial t} = -\bm{J}/\epsilon_{0}, \quad \bm{J}(\bm{x}) = q_{0} \int f\bm{v} ~ \mathrm{d}^{3} \bm{v} ,
\end{equation}
% $\nabla_{\bm{x}} \cdot \bm{E} = \dfrac{\rho - \bar{\rho}}{\epsilon_{0}}$
% $\rho(\bm{x}) = q_{0} \int f ~ \mathrm{d}^{3} \bm{v} , \quad \bar{\rho} = \int \rho ~ \mathrm{d} \bm{x} / V$
where $f(\bm{x},\bm{v},t): \mathbb{R}^{3} \times \mathbb{R}^{3} \times \mathbb{R}_{+} \to \mathbb{R}_{+}$ denotes the single-particle distribution function in the phase space driven by the advection term on the left and the collision part $C[f]$ on the right, $\bm{E}(\bm{x},t)$ is the self-consistent electric field and $\bm{J}$ is the current density.

The VAC model is equivalent to the Vlasov-Poisson-collision model when the charge continuity equation $\rho_{t} + \nabla_{\bm{x}}\cdot \bm{J} = 0$ holds, as the electric field satisfies the Poisson equation $\nabla_{\bm{x}} \cdot \bm{E} = \frac{\rho - \bar{\rho}}{\epsilon_{0}}$.
In both the Amp\`ere and Poisson models, the total energy $\mathcal{E} = \int \frac{1}{2} m |\bm{v}|^{2} f \mathrm{d} \bm{x} \mathrm{d}^{3} \bm{v} + \int \frac{\epsilon_{0}}{2} |\bm{E}|^{2} \mathrm{d} \bm{x}$ is conserved.
The collision operator $C[f]$ represents the dissipative component of the metriplectic dynamics and models unresolved particle interactions beyond the mean-field approximation.
Throughout this work, all particles are assumed to have identical mass $m_{0}$ and charge $q_{0}$, particles interact via an unscreened Coulomb interaction, and the electrons are assumed to be uniformly distributed in space with density $\bar{\rho}$.
% $C[f]$ represents the unresolved particle interactions beyond the mean field approximation. 
% and $(\bm{x},\bm{v})$ represents a phase space point with position.  $\phi(\bm{x} ;f)=\int V(\bm{x},\bm{x}')f(\bm{z}')\mathrm{d}\bm{z}'$ is the mean field potential energy with pair potential $V(\bm{x}, \bm{x}')$.

We introduce characteristic length $L_{0}$, velocity $V_{0}$, particle mass $m_{0}$, charge $q_{0}$, number density per unit of space $n_{0}$, probability density $f_{0}=n_{0}/V_{0}^{3}$, energy $\mathcal{E}_{0} = m_{0} n_{0} L_{0} V_{0}^{2}$, electric field $E_{0} = m_{0} V_{0}^{2} / q_{0} L_{0}$, temperature $T_{0} = m_{0} V_{0}^{2} / k_{B}$, and $\lambda_{D} = [\epsilon_{0} m_{0} V_{0}^{2} / (n_{0} q_{0}^{2} L_{0}^{2})]^{1/2}$ in \ref{app:phys}.
Define $\tilde{\bm{x}} = \bm{x} / L_{0}$, $\tilde{\bm{v}} = \bm{v}/V_{0}$, $\tilde{t} = t V_{0} / L_{0}$, $\tilde{\bm{E}} = \bm{E}/E_{0}$, and $\tilde{f} = f / f_{0}$.
In the rest of the paper, we will drop all the tildes and use the dimensionless equations
\begin{equation}\label{eq:VAC}
    \left\{
    \begin{aligned}
        &\dfrac{\partial f}{\partial t} + \bm{v} \cdot \nabla_{x} f + \bm{E} \cdot \nabla_{\bm{v}} f  = C[f] , \\ % \left.\dfrac{\partial f}{\partial t}\right|_{c}
        &\dfrac{\partial \bm{E}}{\partial t} = - \lambda_{D}^{-2} \bm{J}, \quad \bm{J}(\bm{x}) = \int f\bm{v} ~ \mathrm{d}^{3} \bm{v} ,
    \end{aligned}
    \right.
\end{equation}
where the conserved energy is $\mathcal{E} = \int \frac{1}{2} |\bm{v}|^{2} f \mathrm{d} \bm{x} \mathrm{d}^{3} \bm{v} + \int \frac{\lambda_{D}^{2}}{2} |\bm{E}|^{2} \mathrm{d} \bm{x}$.

The dimensionless VAC system \eqref{eq:VAC} admits a metriplectic structure \cite{morrison1986paradigm} with state variable $\bm z=(f(\bm{x},\bm{v},t), \bm{E}(\bm{x},t))$, where the evolution is decomposed into a Hamiltonian transport part and a dissipative collisional part. Specifically, the Hamiltonian dynamics describes collisionless phase-space transport under a self-consistent field, while the collision operator accounts for irreversible effects arising from unresolved particle interactions. 
For an observable $\mathcal{A}[\bm z]$, the evolution takes the form
\begin{equation}
\dot{\mathcal{A}} = \left\{\mathcal{A}, \mathcal{H}\right\} + \left(\mathcal{A}, \mathcal{S}\right)_+
\label{eq:metriplectic_form}
\end{equation}
where $\mathcal{H}$ is the Hamiltonian functional 
\begin{equation*}
    \mathcal{H}[\bm z] = \int \dfrac{|\bm{v}|^{2}}{2} f(\bm{x},\bm{v}) \mathrm{d} \bm{x} \mathrm{d} \bm{v} + \dfrac{\lambda_{D}^{2}}{2} \int |\bm{E}(\bm{x})|^{2} \mathrm{d} \bm{x},
\label{eq:energy}
\end{equation*}
paired with the Poisson bracket
\begin{equation}\label{eq:poisson_bracket}
\begin{aligned}
    \{ \mathcal{F} , \mathcal{G} \} =& \int f \left( \nabla_{\bm{x}}\dfrac{\delta \mathcal{F}}{\delta f} \cdot \nabla_{\bm{v}}\dfrac{\delta \mathcal{G}}{\delta f}  - \nabla_{\bm{v}}\dfrac{\delta \mathcal{F}}{\delta f} \cdot \nabla_{\bm{x}}\dfrac{\delta \mathcal{G}}{\delta f} \right) \mathrm{d}\bm{x} \mathrm{d} \bm{v} \\
    &+ \lambda_{D}^{-2} \int f \left( \nabla_{\bm{v}} \dfrac{\delta \mathcal{F}}{\delta f} \cdot \dfrac{\delta \mathcal{G}}{\delta \bm{E}} - \nabla_{\bm{v}} \dfrac{\delta \mathcal{G}}{\delta f} \cdot \dfrac{\delta \mathcal{F}}{\delta \bm{E}} \right)
    \mathrm{d}\bm{x} \mathrm{d} \bm{v}  ,
\end{aligned}
\end{equation}
to generate the conservative transport and the Amp\`ere equation. The collisional process is generated by an entropy functional
\begin{equation}
    \mathcal{S}[f] = -\int f \ln f \mathrm{d} \bm{x} \mathrm{d} \bm{v}, \qquad \left\{\mathcal{A}, \mathcal{S}\right\} \equiv 0
\label{eq:entropy}
\end{equation}
paired with a symmetric, positive semi-definite dissipative bracket $(\mathcal{F},\mathcal{G})$ such that the collision operator satisfies
\begin{equation}
    (\mathcal{A}, \mathcal{S})_+ = \int \dfrac{\delta \mathcal{A}}{\delta f} C[f] \mathrm{d} \bm{x} \mathrm{d} \bm{v},  \qquad (\mathcal{H}, \mathcal{S})_+ \equiv 0,
\label{eq:dissipative_bracket}
\end{equation}
where the degeneracy conditions in Eqs. \eqref{eq:entropy},\eqref{eq:dissipative_bracket} ensure the conservation of energy. In particular, by choosing $\mathcal{A} = f$, Eq. \eqref{eq:metriplectic_form} recovers the VAC system \eqref{eq:VAC}.

% in which the time evolution of any observable $\mathcal{A}[f]$ is decomposed into a conservative Hamiltonian part and a dissipative collisional part.
% Specifically, the collisionless Vlasov–Poisson dynamics is generated by a Hamiltonian functional
% \begin{equation*}
%     \mathcal{H}[f] = \int f(\bm{x},\bm{v}) \dfrac{\bm{v}^{2}}{2} \mathrm{d} \bm{x} \mathrm{d} \bm{v} + \dfrac{1}{2} \int |\bm{E}(\bm{x})|^{2} \mathrm{d} \bm{x},
% \end{equation*}
% together with the Poisson bracket
% \begin{equation*}
%     \{ \mathcal{F} , \mathcal{G} \} = \int f \left( \nabla_{\bm{x}}\dfrac{\delta \mathcal{F}}{\delta f} \cdot \nabla_{\bm{v}}\dfrac{\delta \mathcal{G}}{\delta f}  - \nabla_{\bm{v}}\dfrac{\delta \mathcal{F}}{\delta f} \cdot \nabla_{\bm{x}}\dfrac{\delta \mathcal{G}}{\delta f} \right) \mathrm{d}\bm{x} \mathrm{d} \bm{v},
% \end{equation*}
% which generates the conservative transport and self-consistent field coupling.
% The collisional effects are represented by a symmetric, negative semi-definite dissipative bracket $(\mathcal{F},\mathcal{G})$, together with an entropy functional
% \begin{equation*}
%     \mathcal{S}[f] = -\int f \ln f \mathrm{d} \bm{x} \mathrm{d} \bm{v},
% \end{equation*}
% such that the collision operator satisfies
% \begin{equation*}
%     \int \dfrac{\delta \mathcal{A}}{\delta f} C[f] \mathrm{d} \bm{x} \mathrm{d} \bm{v} = (\mathcal{A}, \mathcal{S}), \qquad (\mathcal{H}, \mathcal{S}) = 0,
% \end{equation*}
% ensuring conservation of energy and monotonic entropy production.

To close Eq. \eqref{eq:VAC}, we need to specify the form of the collision operator $C[f]$, or essentially, the dissipative bracket. 
One common choice is the Landau operator 
% Within this metriplectic framework, the classical Landau collision operator corresponds to a particular choice of the dissipative bracket under the assumptions of weak coupling and small-angle binary Coulomb scattering.
% In dimensionless form, the Landau operator for a one-component plasma reads
\begin{equation*}
    C_{L}[f] = \gamma \nabla_{\bm{v}} \cdot \int \dfrac{|\bm{u}|^{2} \bm{I} - \bm{u}\bm{u}^{T}}{|\bm{u}|^{3}} \left( \nabla_{\bm{v}} f(\bm{v}) f(\bm{v}') - \nabla_{\bm{v}'} f(\bm{v}') f(\bm{v}) \right), \qquad \bm{u} = \bm{v}-\bm{v}' ,
\end{equation*}
where the prefactor $\gamma \propto \ln\Lambda$ and the Coulomb parameter scales as $\Lambda \sim \Gamma^{-3/2}$ and $\Gamma = \frac{q_{e}^{2} (4\pi n/3)^{1/3}}{4\pi \epsilon_{0} k_{B} T}$.
In particular, the Landau operator $C_L[f]$  assumes an isotropic binary collision in the velocity space that only depends on the relative velocity $\bm u$ of the colliding particles. Moreover, the spatially inhomogeneous effect is modeled through the empirical form of the Coulomb logarithm $\ln \Lambda$.   
While the Landau operator $C_L[f]$ provides a good approximation in the weak coupling regime $\Gamma \ll 1$, the operator loses the accuracy for the moderately coupled regime $\Gamma \sim \mathcal{O}(1)$ (i.e., $\ln \Lambda \le 0$), which motivates the present data-driven construction of spatially inhomogeneous $C[f]$ directly from the microscale MD simulations.

\subsection{Inhomogeneous metriplectic collision operator}
\label{sec:generalized_kernel}

%Within the metriplectic framework \cite{morrison1986paradigm}, 
% the kinetic evolution of the distribution function is consistently coupled with reversible Hamiltonian (Poisson) dynamics and irreversible dissipative (metric) processes.
% The Vlasov dynamics is generated by a Poisson bracket $\{\cdot,\cdot\}$ and the Hamiltonian functional $\mathcal{H}[f]$, while collisional dissipation is generated by a symmetric dissipative bracket $(\cdot,\cdot)$ acting on the entropy functional $\mathcal{S}[f]$.

% The entropy functional is chosen as
% \begin{equation}
%     \mathcal{S}[f] = - \int f \ln f \, \mathrm{d}\bm{x}\mathrm{d}\bm{v} .
% \end{equation}

% By construction of the Poisson bracket for the Vlasov equation, the entropy is a Casimir invariant,
% \begin{equation}
%     \{ \mathcal{S}, \mathcal{A} \} = 0 \quad \text{for any functional } \mathcal{A}[f],
% \end{equation}
% and therefore remains unchanged under the Hamiltonian evolution.
% This property guarantees that entropy production is exclusively associated with the collisional dynamics.
Within the metriplectic framework \cite{morrison1986paradigm},  we construct the dissipative bracket between two functionals $\mathcal{A}[f]$ and $\mathcal{B}[f]$ as
\begin{equation}
    (\mathcal{A},\mathcal{B})_+ = \dfrac{1}{2} \int \left( \nabla_{\bm{v}} \dfrac{\delta \mathcal{A}}{\delta f} - \nabla_{\bm{v}'} \dfrac{\delta \mathcal{A}}{\delta f'} \right) \cdot \bm{\omega}(\bm{v},\bm{v}';\rho,T) f f' \left( \nabla_{\bm{v}} \dfrac{\delta \mathcal{B}}{\delta f} - \nabla_{\bm{v}'} \dfrac{\delta \mathcal{B}}{\delta f'} \right) \mathrm{d}\bm{v}' \mathrm{d}\bm{v}\mathrm{d}\bm{x} ,
\label{eq:dissipative_bracket_2}
\end{equation}
where $\bm{\omega}(\bm{v}, \bm{v}'; \rho, T): \mathbb{R}^3 \times \mathbb{R}^3 \times \mathbb{R}_{+} \times \mathbb{R}_{+} \to \mathbb{S}^{3}_{+}$ is the collision kernel that depends on the local density $\rho(\bm{x})$ and temperature $T(\bm{x})$ defined by 
\begin{equation*}
    \rho(\bm{x}) = \int f ~ \mathrm{d}^{3} \bm{v}, ~~ \bar{\bm{v}}(\bm{x}) = \int \bm{v} f ~ \mathrm{d}^{3} \bm{v} / \rho(\bm{x}), ~~ T(\bm{x}) = \int \dfrac{(\bm{v} - \bar{\bm{v}}(\bm{x}))^{2}}{2} f ~ \mathrm{d}^{3} \bm{v} / \rho(\bm{x}). 
\end{equation*}
The dissipative bracket \eqref{eq:dissipative_bracket_2} satisfies 
\begin{equation*}
    (\mathcal{A},\mathcal{B})_+ = (\mathcal{B},\mathcal{A})_+, \qquad (\mathcal{S},\mathcal{S})_+ \ge 0,
\end{equation*}
which implies monotonic entropy production. In particular, by choosing entropy $S[f]$ defined by Eq. \eqref{eq:entropy}, the collision operator $C[f]$ can be written in variational form 
\begin{equation}
    C[f] = \left( f, \mathcal{S} \right)_+ = \nabla_{\bm{v}} \cdot \int \bm{\omega}(\bm{v}, \bm{v}'; \rho, T) \left[ f' ~ \nabla_{\bm{v}} f - f ~ \nabla_{\bm{v}'} f' \right] \mathrm{d}\bm{v}'. 
\label{eq:collision_metripletic}
\end{equation}
Eq. \eqref{eq:collision_metripletic} provides the starting point to construct $C[f]$ by learning a proper form of the collision kernel $\bm \omega(\bm v, \bm v'; \rho, T)$. To form a valid collision operator, $C[f]$ needs to further satisfy proper conservation laws 
\begin{equation}
\int C[f] \phi(\bm v) {\rm d}\bm v = 0 \qquad \phi(\bm v) \in {\text {span}} \left \{1, \bm v, \vert \bm v\vert^2\right \},
\label{eq:conservation_law}
\end{equation}
for the mass, momentum and energy, as well as the frame-indifference constraint
\begin{equation}
C[\tilde{f}] = C[f],
\label{eq:frame_indifference}
\end{equation}
for both $\tilde{f}(\bm x, \bm v, t) = f(\bm{U} \bm x, \bm{U}  \bm v, t)$ with $\bm{U} \in {\rm SO}(3)$ and $\tilde{f}(\bm x, \bm v, t) = f(\bm x, -\bm v, t)$.  
This motivates a generalized form of the collision kernel
\begin{equation}
\begin{split} 
    \bm{\omega} (\bm{v}, \bm{v}'; \rho, T) &= \boldsymbol{\mathcal{P}} \left(g_{r}^2 \widetilde{\boldsymbol{r}}\widetilde{\boldsymbol{r}}^T + g_{s}^2 \widetilde{\boldsymbol{s}}\widetilde{\boldsymbol{s}}^T\right) \boldsymbol{\mathcal{P}} \\
    &= g_{1}^{2} |\bm{\mathcal{P}}\bm{r}|^{2} \bm{\mathcal{P}} + (g_{2}^{2}-g_{1}^{2}) \bm{\mathcal{P}} \bm{r} \bm{r}^{T} \bm{\mathcal{P}} ,  
\end{split}\label{eq:kernel_form}    
\end{equation}
where $\boldsymbol{r}=\boldsymbol{v}+\boldsymbol{v}' - 2\bar{\boldsymbol{v}}$, $\boldsymbol{s}=\boldsymbol{u}\times\boldsymbol{r}$ and $\boldsymbol{\mathcal{P}}=\boldsymbol{I}-\boldsymbol{u}\boldsymbol{u}^T/|\boldsymbol{u}|^2$. $\widetilde{\boldsymbol{u}}=\boldsymbol{u}/|\boldsymbol{u}|$, $\widetilde{\boldsymbol{r}}=\boldsymbol{\mathcal{P}}\boldsymbol{r}/|\boldsymbol{\mathcal{P}}\boldsymbol{r}|$ and $\widetilde{\boldsymbol{s}}=\boldsymbol{s}/|\boldsymbol{s}|$ as mutually orthogonal unit vectors, satisfying $\widetilde{\boldsymbol{u}} \widetilde{\boldsymbol{u}}^{T} + \widetilde{\boldsymbol{r}} \widetilde{\boldsymbol{r}}^{T} + \widetilde{\boldsymbol{s}} \widetilde{\boldsymbol{s}}^{T} = \boldsymbol{I}$. $g_r$ and $g_s$ represent the collision magnitude along the two unit directions in the plane orthogonal to the relative velocity $\bm u = \bm v - \bm v'$. Furthermore, we can rewrite the kernel in the form of the isotropic part $\propto \mathcal{P}$ and the anisotropic part $\propto  \mathcal{P} \bm r \bm r^T \mathcal{P}$, where $g_s^2 = g_1^2 \vert \mathcal{P}\boldsymbol{r}\vert^2$ and $g_r^2 = g_2^2 \vert \mathcal{P}\boldsymbol{r}\vert^2$. 

By constructing scalar encoders $g_{\ast}$ that satisfy the rotational, translational, and reflectional invariance,
\begin{equation}
    g_{\ast}(\bm{U}\bm{v}, \bm{U}\bm{v}'; \rho, T) = g_{\ast}(\bm{v}, \bm{v}'; \rho, T), \qquad 
    g_{\ast}(\bm{v}, \bm{v}'; \rho, T) = g_{\ast}(\bm{v}', \bm{v}; \rho, T)
\label{eq:encoder_g}
\end{equation}
for $\ast = 1, 2$ and $\bm{U} \in \rm{SO}(3)$, we can show that the kernel $\bm{\omega}$ satisfies the symmetry and zero projection properties in Eq. \eqref{eq:kernel_conditions} in \ref{app:prop}, and $C[f]$ by Eq. \eqref{eq:collision_metripletic} \eqref{eq:kernel_form} strictly satisfies the conservation laws \eqref{eq:conservation_law} and frame-indifference constraints \eqref{eq:frame_indifference}. Moreover, the encoders $g_{\ast}$ explicitly depend on the local density $\rho$ and temperature $T$, and therefore naturally captures the spatially inhomogeneous energy transfer for various plasma conditions. 

As a special case, by choosing the collision kernel in an isotropic form $\bm{\omega}(\bm{v},\bm{v}') \propto \ln \Lambda ~ |\bm{u}|^{-1}\bm{\mathcal{P}}$, the collision operator $C[f]$ \eqref{eq:collision_metripletic} recovers the standard Landau model for the weakly coupled regime. In this study, we do not take such heuristic assumptions. Rather, the collision kernel \eqref{eq:kernel_form} takes a generalized symmetry-breaking structure which allows for the anisotropic collision magnitude, i.e., $g_1 \neq g_2$, to account for the heterogeneous energy transfer in the plane orthogonal to the $\bm u$. Moreover, $g_\ast$ not only depends on relative velocity $\bm u$ but also the average velocity $\bm r$ and the local thermodynamic states to capture the many-body effects due to the particle correlations in the moderately coupled regime. In particular, we will learn the encoders $g_1$ and $g_2$ directly from the micro-scale MD simulations specified in next subsection. In the remainder of this work, we use ``DDCO'' to denote the present data-driven collision operator.

\subsection{Data-driven construction from molecular dynamics}
To construct the encoders $g_\ast$ satisfying condition \eqref{eq:encoder_g}, one natural choice is to represent $g_\ast$ by the velocity magnitude of the colliding particles, i.e., $g_{\ast}(\vert \bm u\vert, \vert \bm v\vert, \vert \bm v'\vert; \rho, T)$. However, this non-stationary form violates the convolution structure of $C[f]$ in Eq. \eqref{eq:collision_metripletic} that leads to large computational cost. To enable efficient evaluation of the collision operator $C[f]$, we approximate $g_{\ast}$ using a low-rank spectral separation representation, which decomposes the dependence on the individual velocities and their relative velocity.
Specifically, we employ a set of univariate basis functions to separate the velocity magnitude $|\bm{v}|$, $|\bm{v}'|$ and the relative velocity $|\bm{u}|$, yielding
\begin{equation}\label{eq:SS}
    \begin{aligned}
        g_{\ast}(\bm{v}, \bm{v}'; \rho, T) &= \sum_{j'=1}^{J'} \mathcal{L}_{\ast}^{j'}(|\bm{u}|; \rho, T) 
        \left[ \mathcal{M}_{\ast}^{j'}(|\bm{v}|; \rho, T) \mathcal{N}_{\ast}^{j'}(|\bm{v}'|; \rho, T) + \mathcal{N}_{\ast}^{j'}(|\bm{v}|; \rho, T) \mathcal{M}_{\ast}^{j'}(|\bm{v}'|; \rho, T) \right] , \\
        &= \sum_{j=1}^{J} L_{\ast}^{j}(|\bm{u}|; \rho, T) M_{\ast}^{j}(|\bm{v}|; \rho, T) N_{\ast}^{j}(|\bm{v}'|; \rho, T),
    \end{aligned}
\end{equation}
where the second form explicitly enforces permutation symmetry between $\bm{v}$ and $\bm{v}'$, with $J = 2J'$.
The univariate basis functions $\{\mathcal{L}_{\ast}^{j'}, \mathcal{M}_{\ast}^{j'}, \mathcal{N}_{\ast}^{j'}\}_{j'=1}^{J'}$ are represented by neural networks and learned directly from MD data.
% $L_{\ast}^{2j'-1}=L_{\ast}^{2j'}=\mathcal{L}_{\ast}^{j'}$, $M_{\ast}^{2j'-1}=N_{\ast}^{2j'}=\mathcal{M}_{\ast}^{j'}$ and $M_{\ast}^{2j'}=N_{\ast}^{2j'-1}=\mathcal{N}_{\ast}^{j'}$, for $j'=1,\cdots,J'$, and $J=2J'$.

This low-rank decomposition restores a convolution structure in the velocity space, allowing each term of the collision operator in Eq. \eqref{eq:collision_metripletic} to be evaluated via fast Fourier transforms. As a result, the computational complexity is reduced from $\cO(N_{v}^{2})$ to $\cO(JN_{v} \log N_{v})$, which is critical for the high-dimensional 1D-3V kinetic simulations.
Accordingly, Eq. \eqref{eq:collision_metripletic} can be written as the sum of three contributions,
\begin{equation}\label{eq:sim_C}
    \begin{aligned}
        C[f] =& \nabla \cdot \int g_{1}^{2} |\bm{\mathcal{P}}\bm{r}|^{2} \bm{\mathcal{P}} f(\bm{v}) f(\bm{v}') \left[ \nabla \log f(\bm{v}) - \nabla'\log f(\bm{v}') \right] \mathrm{d}\bm{v}' \\
        &+ \nabla \cdot \int g_{2}^{2} \bm{\mathcal{P}} \bm{r} \bm{r}^{T} \bm{\mathcal{P}} f(\bm{v}) f(\bm{v}') \left[ \nabla \log f(\bm{v}) - \nabla'\log f(\bm{v}') \right] \mathrm{d}\bm{v}' \\
        &- \nabla \cdot \int g_{1}^{2} \bm{\mathcal{P}} \bm{r} \bm{r}^{T} \bm{\mathcal{P}} f(\bm{v}) f(\bm{v}') \left[ \nabla \log f(\bm{v}) - \nabla'\log f(\bm{v}') \right] \mathrm{d}\bm{v}' \\
        =&: I_{1} + I_{2} - I_{3},
    \end{aligned}
\end{equation}
and each term can be further expressed as a finite sum of convolution operators in velocity space, see Eq. \eqref{eq:sim_I1} in \ref{app:prop}.

% The univariate basis functions capture the anisotropic nature of the different energy transfer mechanisms during particle collisions in systems at different temperatures and densities, directly learned by matching the prediction of the collision operator and the micro-scale MD trajectories.
The collision kernel is learned directly from microscopic MD data conditioned on local environment $(\rho,T)$ across the weakly and moderately coupled regimes with the coupling parameter $\Gamma$ up to $\cO(1)$.
The loss function is constructed in a weak form by matching the MD with the kinetic equation at collection of density-temperature pairs $\{\rho^{(l)}, T^{(l)}\}_{l=1}^{N_{l}}$, using a set of test functions $\left \{\psi_k(\bm{v})\right\}_{k=1}^{K}$,
% To learn the low-rank kernel functions from particle data in inhomogeneous cases, we construct a weak-form loss using a set of test functions $\left \{\psi_k(\bm{v})\right\}_{k=1}^K$, by matching MD and the kinetic equation at several pairs of fixed densities $\rho$ and temperatures $T$
\begin{equation*}
    \mathcal{L}= \sum_{l=1}^{N_l} \sum_{k=1}^K \sum_{n=1}^{N_T} \left(\left.\dfrac{\partial f^{n} (\rho^{(l)}, T^{(l)})}{\partial t}\right|_{MD} - C[f]^{n} (\rho^{(l)}, T^{(l)}), ~ \psi_k(\bm{v})\right)^2 ,
\end{equation*}
where $(\cdot, \cdot)$ denotes the inner product in the velocity space, and $N_T$ is the number of temporal snapshots of MD trajectories. 
The velocity distribution is approximate by empirical measure from MD data as $f^{\ast}(\bm{v}; \rho^{(l)}, T^{(l)}) = \frac{1}{N_{MD}}\sum_{m=1}^{N_{MD}} \delta(\bm{v}-\bm{v}_m)$, avoiding explicit density reconstruction.
The MD trajectories employed for training correspond to spatially homogeneous systems in Eq. \eqref{eq:VAC}, as $\tilde{f}(\bm{x},\bm{v},t) = f(\bm{v},t)$, maintaining density and temperature during the MD simulation process.
By training the collision operator on homogeneous MD trajectories across multiple $(\rho^{(l)}, T^{(l)})$, the learned collision operator captures the intrinsic collisional physics at the microscopic level, while spatial inhomogeneity is incorporated at the kinetic level during subsequent simulations.

The MD contribution to the loss is evaluated using finite differences in time,
\begin{equation}
    \left(\left.\dfrac{\partial f^n}{\partial t}\right|_{MD}, \psi_k(\bm{v})\right) \approx
    % \left(\dfrac{f^{n+1}-f^{n}}{\Delta t},\psi_k(\bm{v})\right) =
    \dfrac{1}{N_{MD}\Delta t}\sum_{m=1}^{N_{MD}} [\psi_k(\bm{v}_{m}^{n+1})-\psi_k(\bm{v}_{m}^{n})] ,
\end{equation}
while the kinetic contribution is computed in the weak form,
\begin{equation}
    \left(C[f],\psi_k(\boldsymbol{v})\right) 
    = \left(f(\boldsymbol{v})\int \boldsymbol{\omega} f(\boldsymbol{v}')\mathrm{d}\boldsymbol{v}', \nabla_{\boldsymbol{v}} \nabla_{\boldsymbol{v}} \psi_k(\boldsymbol{v})\right) 
    + \left( f(\boldsymbol{v}) \int (\nabla_{\boldsymbol{v}} - \nabla_{\boldsymbol{v}'}) \cdot \boldsymbol{\omega} ~  f(\boldsymbol{v}') \mathrm{d}\boldsymbol{v}', \nabla_{\boldsymbol{v}} \psi_k (\boldsymbol{v}) \right) .
\end{equation}
Due to the large number of MD particles $N_{MD} \sim 10^{6}$, the velocity integrals are evaluated using random mini-batch sampling \cite{ketkar2017stochastic, li2014efficient, jin2020random, jin2021random},
\begin{equation}
    \begin{aligned}
        \left( f(\boldsymbol{v})\int \boldsymbol{\omega} f(\boldsymbol{v}')\mathrm{d}\boldsymbol{v}', \nabla_{\boldsymbol{v}} \nabla_{\boldsymbol{v}} \psi_k(\boldsymbol{v}) \right) &=
        \dfrac{1}{N_{MD}^2}\sum_{m,m'}^{N_{MD}} \boldsymbol{\omega}(\boldsymbol{v}_{m},\boldsymbol{v}_{m'} ') : \nabla_{\boldsymbol{v}} \nabla_{\boldsymbol{v}} \psi_k(\boldsymbol{v}_{m}), \\
        &\approx \dfrac{1}{P}\sum_{p=1}^{P} \boldsymbol{\omega}(\boldsymbol{v}_{m(p)},\boldsymbol{v}_{m'(p)} ') : \nabla_{\boldsymbol{v}} \nabla_{\boldsymbol{v}} \psi_k(\boldsymbol{v}_{m(p)}),
    \end{aligned}
\end{equation}
with pairs $(\bm{v}_{m(p)}, \bm{v}_{m'(p)} ')$ randomly selected from the MD samples for $p=1,\cdots,P$.

We refer to  \ref{app:MD} for the details on the MD setup, sampling strategy, and training procedure.

\subsection{Numerical scheme}
\label{sec:numerical_scheme}

We consider the one-dimensional, three-velocity (1D-3V) VAC system
\begin{equation}
    \partial_{t} f + v_{x} \nabla_{x} f + E(x,t)\,\nabla_{v_{x}} f = C[f], \qquad (x,\bm{v}) \in \Omega_{x} \times \Omega_{\bm{v}},
\end{equation}
where $\bm{v} = (v_{x}, v_{y}, v_{z})$ and the electric field is updated by $\partial_{t} E(x,t) = - \lambda_{D}^{-2} J(x,t)$ with initial condition $\nabla_{x} E(x,t=0) = \lambda_{D}^{-2} (\rho(x,t=0) - \bar{\rho}(t=0))$, and the density is defined as $\rho(x,t) = \int_{\Omega_{\bm{v}}} f(x,\bm{v},t)\,\mathrm{d}^{3}\bm{v}$.
Here we assume that the velocity space $\Omega_{v}$ is finitely truncated in practice.
%, which is reasonable because the solutions remain compactly supported with certain conditions.
In all numerical schemes and results below, we assume that the global solutions exist and $\Omega_{v}$ is large enough, so that the numerical solutions satisfy $f\approx0$ and $f v^{2}\approx 0$ at $\partial \Omega_{v}$.

The spatial domain $\Omega_{x}=[0,L_{x}]$ and the velocity domain $\Omega_{\bm{v}}=[-V_{x},V_{x}]\times[-V_{y},V_{y}]\times[-V_{z},V_{z}]$ are discretized using uniform grids,
\begin{equation}
    \begin{aligned}
        &x_{i} = \left(i+\dfrac{1}{2}\right)\Delta x, \qquad i=0,\ldots,N_{x}-1, \qquad \Delta x = \dfrac{L_{x}}{N_{x}}, \\
        &v_{j_{\alpha}} = -V_{\alpha} + \left(j_{\alpha}+\dfrac{1}{2}\right)\Delta v_{\alpha}, \qquad j_{\alpha}=0,\ldots,N_{v_{\alpha}}-1, \qquad \Delta v_{\alpha} = \dfrac{2V_{\alpha}}{N_{v_{\alpha}}}, \quad \alpha\in\{x,y,z\}.
    \end{aligned}
\end{equation}

The distribution function is approximated by its cell-averaged values
\begin{equation}
    f^{n}_{i,\bm{j}} \approx f(x_{i},v_{j_{x}},v_{j_{y}},v_{j_{z}},t^{n}),
\end{equation}
where $\bm{j}=(j_{x},j_{y},j_{z})$, $\Delta v = \Delta v_{x}\Delta v_{y}\Delta v_{z}$, and the discrete mass, momentum, kinetic energy, electric potential energy, total energy, and entropy are defined as
\begin{equation}\label{eq:MPE}
    \begin{aligned}
        M^{n} =& \Delta x \Delta v \sum_{i,\bm{j}} f_{i,\bm{j}}^{n}, ~~ &\bm{P}^{n} =& \Delta x \Delta v \sum_{i,\bm{j}} \bm{v}_{\bm{j}} f_{i,\bm{j}}^{n}, \\
        \mathcal{E}_{K}^{n} =& \Delta x \Delta v \sum_{i,\bm{j}} \dfrac{|\bm{v}_{\bm{j}}|^{2}}{2} f_{i,\bm{j}}^{n}, ~~ &\mathcal{E}_{P}^{n} =& \Delta x \sum_{i} \dfrac{\lambda_{D}^{2}}{2} |E_{i}^{n}|^{2} , \\
        \mathcal{E}^{n} =& \mathcal{E}_{K}^{n} + \mathcal{E}_{P}^{n}, ~~ &S^{n} =& - \Delta x \Delta v \sum_{i,\bm{j}} f_{i,\bm{j}}^{n} \log f_{i,\bm{j}}^{n} .
    \end{aligned}
\end{equation}

In this work, we develop an explicit second-order energy-conservation numerical scheme.
The evolution operator is decomposed into advection and collision terms
\begin{equation}
    \partial_{t} f = \mathcal{F}f, \qquad \mathcal{F} = \mathcal{F}_{a} + \mathcal{F}_{C},
\end{equation}
with
\begin{equation}
    \mathcal{F}_{a}f = -v_{x}\partial_{x}f -E(x,t)\nabla_{v_{x}}f, \qquad \mathcal{F}_{C}f = C[f].
\end{equation}

To advance the solution in time, we employ a second-order Strang splitting
\cite{strang1968construction}. 
For a time step $\Delta t$, the single-step propagator is approximated by
\begin{equation}
    e^{\Delta t\,\mathcal{F}} = e^{\frac{\Delta t}{2}\mathcal{F}_{a}} e^{\Delta t\,\mathcal{F}_{C}} e^{\frac{\Delta t}{2}\mathcal{F}_{a}} + \mathcal{O}(\Delta t^{3}).
\end{equation}

\subsubsection{Advection}
We use the explicit second-order scheme \cite{cheng2014energy} of the advection term $\partial_{t} f + v_{x} \nabla_{x} f + E(x,t)\,\nabla_{v_{x}} f = 0$ with Amp\`ere solver  for a single time step $\Delta t$,
\begin{equation}\label{eq:VAadv}
\begin{aligned}
    &\dfrac{f^{n+1/2} - f^{n}}{\Delta t / 2} + v_{x} \cdot \nabla_{x} f^{n} + E^{n} \cdot \nabla_{v_{x}} f^{n} = 0 , \\
    &\dfrac{E^{n+1}-E^{n}}{\Delta t} = - \lambda_{D}^{-2} J^{n+1/2}, ~~ \text{where} ~ J^{n+1/2} = \int f^{n+1/2} \bm{v} ~ \mathrm{d}\bm{v} , \\
    &\dfrac{f^{n+1}-f^{n}}{\Delta t} + v_{x} \cdot \nabla_{x} f^{n+1/2} + \dfrac{1}{2} (E^{n}+E^{n+1}) \cdot \nabla_{v_{x}} f^{n+1/2} = 0 ,
\end{aligned}
\end{equation}
where $E^{n}$ and flux $J^{n+1/2}$ are defined on the spatial grid points.

We take the upwind scheme for the spatial-advection term
\begin{equation}
\begin{aligned}
    v_{x} \cdot \nabla_{x} f_{i,\bm{j}}^{n} =& \dfrac{1}{\Delta x} \left( F_{i+1/2,\bm{j}} - F_{i-1/2,\bm{j}} \right), \\
    F_{i+1/2,\bm{j}} =& v_{j_x} f_{i+1/2,\bm{j}}^{\mathrm{up}} , \\
    f_{i+1/2,\bm{j}}^{\mathrm{up}} =&
    \begin{cases}
        f_{i,\bm{j}}, & v_{j_x} > 0, \\
        f_{i+1,\bm{j}}, & v_{j_x} < 0 ,
    \end{cases}
\end{aligned}
\end{equation}
which preserves the discrete mass, momentum, and kinetic energy exactly.
% and satisfies a discrete entropy inequality under the CFL condition $\frac{\Delta t}{\Delta x}\max |v_x|\le 1$.
% This monotone discretization provides robustness for long-time transport and guarantees mass conservation and nonlinear stability in the spatial direction.

The velocity-advection is discretized componentwise using a conservative central flux,
\begin{equation}
\begin{aligned}
    E_{i}^{n} \cdot \nabla_{v_{x}} f_{i,\bm{j}}^{n} =& \dfrac{1}{\Delta v_x} \left( G_{i,\bm j+1/2_x} - G_{i,\bm j-1/2_x} \right) , \\
    G_{i,\bm j+1/2_x} =& E_i \dfrac{f_{i,\bm j} + f_{i,\bm j+1_x}}{2} ,
\end{aligned}
\end{equation}
where $\bm j=(j_x,j_y,j_z)$, $\bm j+1_x=(j_x+1,j_y,j_z)$, and $\bm j+1/2_x=(j_x+1/2,j_y,j_z)$.

\begin{proposition}
    The Amp\`ere solver iteratively updates the electric field $E^{n}$ instead of recalculating it from the density distribution, which ensures the conservation of total energy $\mathcal{E}$.
\end{proposition}
\begin{proof}
    The total energy is defined by
    \begin{equation}
        \mathcal{E}^{n} = \Delta x \Delta v \sum_{i,\bm{j}} \dfrac{\bm{v}_{\bm{j}}^{2}}{2} f_{i,\bm{j}}^{n} + \Delta x \sum_{i} \dfrac{\lambda_{D}^{2}}{2} |E_{i}^{n}|^{2}
    \end{equation}
    
    So that the energy difference
    \begin{equation}
    \begin{aligned}
        \dfrac{2(\mathcal{E}^{n+1}-\mathcal{E}^{n})}{\Delta t} =& - \Delta x \Delta v \sum_{i,\bm{j}} \bm{v}_{\bm{j}}^{2} 
        \left(\dfrac{1}{\Delta x} ( F_{i+1/2,\bm{j}}^{n+1/2} - F_{i-1/2,\bm{j}}^{n+1/2} ) + \dfrac{1}{4\Delta v_{x}} (E_{i}^{n}+E_{i}^{n+1}) \cdot (f_{i,\bm j+1_x}^{n+1/2} - f_{i,\bm j-1_x}^{n+1/2}) \right)\\
        & - \Delta x \sum_{i} (E_{i}^{n}+E_{i}^{n+1}) \cdot J_{i}^{n+1/2} \\
        =& - \Delta v \sum_{\bm{j}} \bm{v}_{\bm{j}}^{2} \left( \sum_{i} (F_{i+1/2,\bm{j}}^{n+1/2} - F_{i-1/2,\bm{j}}^{n+1/2}) \right) 
        - \dfrac{\Delta x \Delta v}{4 \Delta v_{x}} \sum_{i} (E_{i}^{n}+E_{i}^{n+1}) \sum_{\bm{j}} \bm{v}_{\bm{j}}^{2} (f_{i,\bm j+1_x}^{n+1/2} - f_{i,\bm j-1_x}^{n+1/2}) \\
        & - \Delta x \sum_{i} (E_{i}^{n}+E_{i}^{n+1}) \cdot J_{i}^{n+1/2} \\
        =& \Delta x \sum_{i} (E_{i}^{n}+E_{i}^{n+1}) \Delta v \sum_{\bm{j}} \bm{v}_{\bm{j}} f_{i,\bm j}^{n+1/2}
        - \Delta x \sum_{i} (E_{i}^{n}+E_{i}^{n+1}) \cdot J_{i}^{n+1/2} \\
        =& 0 ,
    \end{aligned}
    \end{equation}
    where the first term in line 3 vanishes due to the periodic boundary conditions in the x direction, and the third equation holds for $\sum_{\bm{j}} \bm{v}_{\bm{j}}^{2} (f_{i,\bm j+1_x}^{n+1/2} - f_{i,\bm j-1_x}^{n+1/2}) = - \sum_{\bm{j}} (\bm{v}_{\bm{j}-1_{x}}^{2} - \bm{v}_{\bm{j}+1_{x}}^{2}) f_{i,\bm j}^{n+1/2}$
    % % continuous
    % \begin{equation}
    %     \mathcal{E}^{n} = \dfrac{1}{2} \int f^{n} |\bm{v}|^{2} ~ \mathrm{d}\bm{v}\mathrm{d} x + \dfrac{1}{2} \int |E^{n}|^{2} ~ \mathrm{d} x .
    % \end{equation}
    % \begin{equation}
    % \begin{aligned}
    %     2(\mathcal{E}^{n+1}-\mathcal{E}^{n}) =& -\int \bm{v} \cdot \nabla_{x} f^{n+1/2} |\bm{v}|^{2} + \dfrac{1}{2} (E^{n}+E^{n+1}) \cdot \nabla_{v_{x}} f^{n+1/2} |\bm{v}|^{2} ~ \mathrm{d}\bm{v}\mathrm{d} x \\
    %     & - \int (E^{n+1}+E^{n})\cdot J^{n+1/2} ~ \mathrm{d}x \\
    %     =& -\int |\bm{v}|^{2}\bm{v} \cdot \left(\int \nabla_{x} f^{n+1/2} ~ \mathrm{d}x\right) ~ \mathrm{d}\bm{v} + \int (E^{n}+E^{n+1}) \cdot (f^{n+1/2}\bm{v}) ~ \mathrm{d}\bm{v}\mathrm{d} x \\
    %     & - \int (E^{n+1}+E^{n})\cdot J^{n+1/2} ~ \mathrm{d}x \\
    %     =& 0 .
    % \end{aligned}
    % \end{equation}
\end{proof}

\begin{remark}
    The numerical scheme in Eq. \eqref{eq:VAadv} only conserves the total mass and energy, while it cannot conserve the total momentum. Alternatively, we propose a momentum conservation scheme in \ref{app:mom_conserv_scheme} with a first-order time discretization error for the total energy.
\end{remark}

% This discretization satisfies a discrete summation-by-parts property, which ensures exact conservation of mass and momentum in velocity space at the semi-discrete level. 
% Periodic boundary conditions are imposed in all velocity directions. 
% To avoid spurious interactions induced by velocity-space periodicity, the computational velocity domain in each direction is extended by zero padding whose width is at least comparable to the effective nonzero support of the distribution function.

\subsubsection{Collision}
The second-order discretization of the collision term $\partial_{t} f = C[f]$ is 
\begin{equation}
\begin{aligned}
    \dfrac{f^{n+1/2} - f^{n}}{\Delta t/2} =& C[f^{n}] , \\
    \dfrac{f^{n+1} - f^{n}}{\Delta t} =& C[f^{n+1/2}] ,
\end{aligned}
\end{equation}
with 
\begin{equation}
    \begin{aligned}
        &C[f](x,\bm{v},t) = \nabla \cdot \int \bm{\omega}(\bm{v}, \bm{v}'; \rho, T) \left[ f(\bm{v}') \nabla f(\bm{v}) - f(\bm{v}) \nabla' f(\bm{v}') \right] \mathrm{d}\bm{v}',  \quad (\bm{x},\bm{v}) \in \Omega_{x} \times \Omega_{v} , \\
        &\bm{\omega}(\bm{v}, \bm{v}'; \rho, T) = g_{1}^{2} |\bm{\mathcal{P}}\bm{r}|^{2} \bm{\mathcal{P}} + (g_{2}^{2}-g_{1}^{2}) \bm{\mathcal{P}} \bm{r} \bm{r}^{T} \bm{\mathcal{P}} , \\
        &\left( \int \bm{\omega}(\bm{v}, \bm{v}'; \rho, T) \left[ f(\bm{v}') \nabla f(\bm{v}) - f(\bm{v}) \nabla' f(\bm{v}') \right] \mathrm{d}\bm{v}'\right) \cdot \hat{\bm{n}}(\bm{v}) = 0, ~~\quad (\bm{x},\bm{v}) \in \Omega_{x} \times \partial\Omega_{v}.
    \end{aligned}
\end{equation}
with generalized collision kernel further depends parametrically on the local density and temperature, and structure-preserving encoder functions $g_{\ast}(\bm{v}, \bm{v}'; \rho, T)$ directly learned from MD data.
% constructed by minimizing the empirical loss function
% retains the MD fidelity beyond empirical models, and meanwhile, strictly preserves the physical constraints

At each spatial grid point $\bm{x}_{i}$, the collision operator is discretized in velocity space using a conservative spectral method.
The generalized collision kernel depends on the local density and temperature, which are computed consistently from the distribution function.
% construct the semi-discrete structure-preserving scheme based on the finite-difference method
\begin{equation}
    \begin{aligned}
        \dfrac{\partial f_{\bm{j}}(t)}{\partial t} =& \mathrm{D}^{+} \bm{p}_{\bm{j}}, \qquad 
        \mathrm{D}^{+} \bm{p}_{\bm{j}} = \sum_{\alpha} [\bm{p}_{\bm{j}+1_{\alpha}} - \bm{p}_{\bm{j}-1_{\alpha}}] / 2\Delta v_{\alpha} , \\
        % \mathrm{D}^{+} \bm{p}_{\bm{j}} , \\ % k,l,p
        \bm{p}_{\bm{j}} =& \Delta v \sum_{\bm{j}' \in \mathbb{Z}^{3}} \bm{\omega}(\bm{v}_{\bm{j}}, \bm{v}_{\bm{j}'}; \rho, T) f_{\bm{j}} f_{\bm{j}'} \left[ \mathrm{D}^{-} \log f_{\bm{j}} - \mathrm{D}^{-} \log f_{\bm{j}'} \right] , \\
        \rho =& \Delta v \sum_{\bm{j}} f_{\bm{j}}, ~~ \bar{\bm{v}}(x) = \Delta v \sum_{\bm{j}} \bm{v}_{j} f_{\bm{j}} ~ / ~ \rho, ~~ T = \Delta v \sum_{\bm{j}} \dfrac{\bm{v}_{j}^{2}}{2} f_{\bm{j}} ~ / ~ \rho ,
    \end{aligned}
\end{equation}
where $\mathrm{D}^{-}$ and $\mathrm{D}^{+}$ are dual discrete gradient and divergence operators of central difference discretization, the subscript $\bm{j} = (j_{x}, j_{y}, j_{z}) \in \mathbb{Z}^{3}$, $\bm{j}+1_{x} = (j_{x}+1, j_{y}, j_{z})$, $\Delta v = \Delta v_{x} \Delta v_{y} \Delta v_{z}$, $\bm{p}_{\bm{j}}$ denotes the discrete probability flux at $\bm{v}_{\bm{j}}$.
The resulting semi-discrete scheme satisfies exact conservation of mass and momentum, conservation of energy up to time discretization errors, and a discrete H-theorem pointwise in space.
The discrete fast Fourier transform is applied to reduce the computational complexity from $\cO(N_v^2)$ to $\cO(J'^{2} N_v \log N_v)$.
The detailed derivation of the velocity-space discretization can be found in Ref. \cite{zhao2025fast}.

% $\mathrm{D}^{-}$ and $\mathrm{D}^{+}$ are dual discrete gradient and divergence operators of central difference discretization
% central difference scheme for simulation

\begin{proposition}
    The numerical semi-discretization scheme defined above satisfies the conservation law of mass, momentum, and kinetic energy, and the H-theorem in the discrete sense.
\end{proposition}

\begin{proof}
For a fixed spatial grid point $x_i$, any discrete physical quantity and its semi-discrete evolution can be written in the form
\begin{equation}
    \begin{aligned}
        \Phi &= \Delta v \sum_{\bm{j}} \phi_{\bm{j}} f_{\bm{j}}, \\
        \dfrac{\mathrm d \Phi}{\mathrm d t}
        &= -\dfrac{1}{2}\Delta v^{2}
        \sum_{\bm{j},\bm{j}'}
        \bigl(\mathrm{D}^{-}\phi_{\bm{j}}
        - \mathrm{D}^{-}\phi_{\bm{j}'}\bigr)^{T}
        \bm{\omega}(\bm{v}_{\bm{j}},\bm{v}_{\bm{j}'};\rho,T)\,
        f_{\bm{j}} f_{\bm{j}'}
        \bigl(\mathrm{D}^{-}\log f_{\bm{j}}
        - \mathrm{D}^{-}\log f_{\bm{j}'}\bigr),
    \end{aligned}
\end{equation}
which follows directly from the conservative discretization and a discrete summation-by-parts argument.
By choosing the test functions $\phi_{\bm{j}} = \{1, ~ \bm{v}_{j},~ \bm{v}_{\bm{j}}^{2} \}$, we obtain $\mathrm{D}^{-}\phi_{\bm{j}} - \mathrm{D}^{-}\phi_{\bm{j}'} \in \text{Ker}\bigl(\bm{\omega}(\bm{v}_{\bm{j}},\bm{v}_{\bm{j}'};\rho,T)\bigr)$, which implies that the discrete mass, momentum, and kinetic energy are exactly conserved.

For the discrete entropy, choosing $\phi_{\bm{j}}=-\log f_{\bm{j}}$ yields $\frac{\mathrm d \mathcal S}{\mathrm d t} \geq 0$ since the collision kernel $\bm{\omega}$ is symmetric and positive semi-definite.
The equality holds if and only if $\mathrm{D}^{-}\log f_{\bm{j}} - \mathrm{D}^{-}\log f_{\bm{j}'} \in \text{Ker}\bigl(\bm{\omega}(\bm{v}_{\bm{j}},\bm{v}_{\bm{j}'};\rho,T)\bigr)$, which characterizes the discrete local Maxwellian distribution.
\end{proof}

% \begin{equation}
%     \frac{\mathrm d \mathcal S}{\mathrm d t}
%     = \frac{1}{2}\Delta v^{2}
%     \sum_{\bm{j},\bm{j}'}
%     \bigl(\mathrm{D}^{-}\log f_{\bm{j}}
%     - \mathrm{D}^{-}\log f_{\bm{j}'}\bigr)^{T}
%     \bm{\omega}(\bm{v}_{\bm{j}},\bm{v}_{\bm{j}'};\rho,T)\,
%     f_{\bm{j}} f_{\bm{j}'}
%     \bigl(\mathrm{D}^{-}\log f_{\bm{j}}
%     - \mathrm{D}^{-}\log f_{\bm{j}'}\bigr)
%     \ge 0,
% \end{equation}

\section{Numerical results}\label{sec:numerical}

First, let us examine the transport coefficients predicted from the present DDCO over a broad range of plasma conditions. Specifically, we consider the self-diffusion coefficient and shear viscosity  under various density and temperature values. Physically, these transport coefficients characterize the relaxation of the plasma kinetics under perturbation near the Maxwellian distribution, and are determined by the collisional operator under the second-order Chapman-Enskog expansion \cite{chapman1990mathematical,reichl2009modern}, i.e., 
% We calculate the transport coefficients of the present DDCO model using the second-order Chapman-Enskog expansion \cite{chapman1990mathematical,reichl2009modern}, with the self-diffusion coefficient and shear viscosity defined by
\begin{equation}
\begin{aligned}
    D =& -\int v_x \zeta (\bm v) f_M(\bm v) d^{3} \bm v, \qquad  \eta = -\frac{nm^2}{k_BT} \int  f_M(\bm v) v_x v_y \xi(\bm v)  d^{3} \bm v, \\
    C^{-}(h) :=& f_M^{-1}  \nabla \cdot \int \bm \omega(\bm v, \bm v'; \rho, T) \nabla_v(h) f_M(\bm v) f_M (\bm v') d^3\bm v' , \\
    C^{+}(h) :=& f_M^{-1}  \nabla \cdot \int \bm \omega(\bm v, \bm v'; \rho, T) (\nabla_v h(\bm v)- \nabla_{v'}h(\bm v')) f_M(\bm v) f_M (\bm v') d^3\bm v'.
\end{aligned}
\end{equation}
where $\zeta(\bm v) = \sum_{n=0}^p c_n S^n_{3/2} (\vert \bm v\vert) v_x$ and $\zeta(\bm v) = \sum_{n=0}^p b_n S^n_{5/2} (\vert \bm v\vert) v_x v_y$ are determined by expansion of the Sonine polynomial basis functions, satisfying $C^{-} (\zeta(\bm v)) = v_x$ and $C^{+} (\xi(\bm v)) = v_x v_y$. $f_M(\bm v; \rho, T)$ represents the Maxwellian distribution under the density $\rho$ and temperature $T$. Fig. \ref{fig:diff_visc} shows the prediction of the transport coefficients predicted from the present DDCO and the Landau model, as well as the values obtained from the direct micro-scale full MD simulations. For the weakly coupled regime $\Gamma = \frac{q_{e}^{2} (4\pi n/3)^{1/3}}{4\pi \epsilon_{0} k_{B} T} \ll 1$ (i.e., low density, high-temperature), the prediction from both Landau and the present DDCO show good agreement with the MD results. However, for moderately coupled regime, the prediction of the Landau model show apparent deviation. In contrast, the present DDCO model can still accurately characterize the macroscopic phenomena of the system.

% It shows that at high temperatures, the diffusion and viscosity of the system can be captured by the Landau equation, but at low temperatures, i.e., under strong coupling conditions, the Landau equation fails. 
% However, our model DDCO can still accurately characterize the macroscopic phenomena of the system.

\begin{figure}[H]
    \centering
    \includegraphics[width=0.7\textwidth]{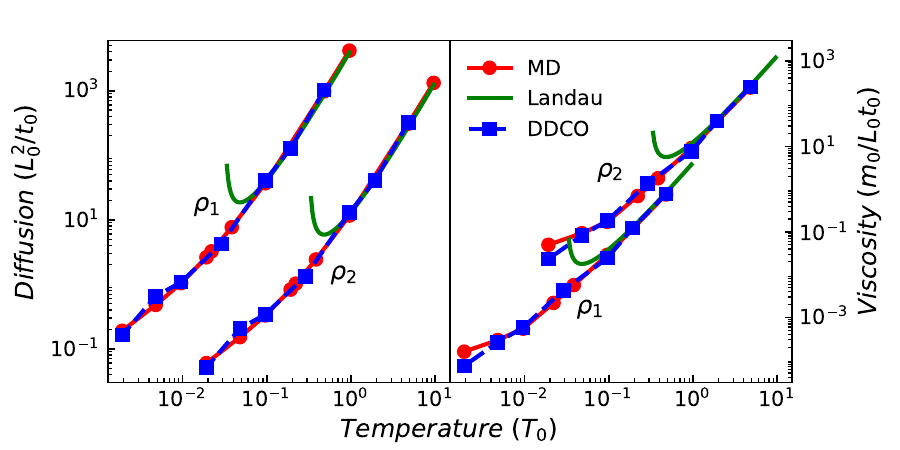}
    \caption{Comparison of the diffusion and viscosity coefficients predicted from MD, Landau equation and DDCO models under different temperatures $T_0$ and densities $\rho = m_0 n$ with number density $n_{1}=10^{21}~{\rm m}^{-3}$ and $n_{2}=10^{24}~{\rm m}^{-3}$.}
    \label{fig:diff_visc}
\end{figure}

Next, we consider the 1D-3V plasma kinetics with the initial condition
\begin{equation*}
f(x, \bm v, t=0) = \rho_0  \tilde{f}(\bm{v}_{\bm{j}}; x_{i}) 
\end{equation*}
where $\rho_0 = n_0m_0$ with $n_0 = 10^{24} m^{-3}$ and $\tilde{f}(\bm{v}_{\bm{j}}; x_{i})$ takes various distributions with spatially dependent temperature, i.e.,  
\begin{equation}
\begin{aligned}
    &\int  \tilde{f}(\bm{v}; x) {\rm d} \bm v = 1 , 
    \quad \quad \qquad \quad \int \bm{v}  \tilde{f}(\bm{v}; x) {\rm d} \bm v = \bm 0, \\
    &\int \vert \bm{v} \vert^{2} \tilde{f}(\bm{v}; x) {\rm d} = T(x) , \quad~~
    T(x) = 0.2 + 0.1 \sin(2\pi/L_{x}) , 
\end{aligned}    
\end{equation}
where $0<x<L_{x}$, and the unit of temperature is taken as \text{eV}.

% the x-axis is divided into $20$ grids, and the initial PDF at grids is 
% \begin{equation}
% \begin{aligned}
%     &f(x_{i},\bm{v}_{\bm{j}}) = \rho(x_{i}) \tilde{f}(\bm{v}_{\bm{j}}; x_{i}) , \\
%     &\Delta v \sum_{\bm{j}} \tilde{f}(\bm{v}_{\bm{j}}; x_{i}) = 1 , \\
%     &\Delta v \sum_{\bm{j}} \bm{v}_{\bm{j}} \tilde{f}(\bm{v}_{\bm{j}}; x_{i}) = \bar{\bm{v}}(x_{i}) , \\
%     &\Delta v \sum_{\bm{j}} [\bm{v}_{\bm{j}} - \bar{\bm{v}}(x_{i})]^{2} \tilde{f}(\bm{v}_{\bm{j}}; x_{i}) = T(x_{i}) , \\
%     &T(x_{i}) = 0.2 + 0.1 \sin[(2n-1)\pi/20] ~ \text{eV} ,
% \end{aligned}    
% \end{equation}
% where the distribution function has the same shape in the x-direction, but the temperature differs at different x-grid points (sine function).

In this work, we choose $\tilde{f}(\bm{v}_{\bm{j}}; x_{i})$ as bi-Maxwellian, symmetric double-well, asymmetric double-well. We refer to  \ref{app:MD} to the detailed form of $\tilde{f}$ and the MD setup. The plasma kinetics beyond the weakly coupled regime where coupling parameter $\Gamma$ takes the value between $0.75$ and $2.3$. The Landau model is insufficient for such regime as the prefactor $\ln \Lambda$ got negative value, which can also be seen in Fig. \ref{fig:diff_visc}. 
% It is important to note that if the classic Landau model is used, the prefactor in the Landau model $\gamma \propto \ln \Lambda$ becomes negative, rendering the simulation results unreliable, as detailed in the Appendix.
% Therefore, we estimate the dissipation rate of the Landau simulations under different $(\rho,T)$ environments to obtain an approximate prefactor function, which is used as the comparison result with the classic Landau simulation.

In the numerical simulations, we take $L_{x} = 1024~\AA = 10.24 ~ L_{0}$, $V_{x} = V_{y} = V_{z} = 8 ~ V_{0}$, with $N_{x} = 20$, $N_{v_x} = N_{v_y} = N_{v_z} = 200$, $\Delta x = 0.512 ~ L_{0}$, $\Delta v_{x} = \Delta v_{y} = \Delta v_{z} = 0.08 ~ V_{0}$, and the time step $\Delta t = 0.02 ~ t_{0}$, satisfying the CFL condition $\frac{\Delta t}{\Delta x}\max |v_x|\le 1$.

\begin{figure}[H]
    \centering
    \includegraphics[width=0.7\textwidth]{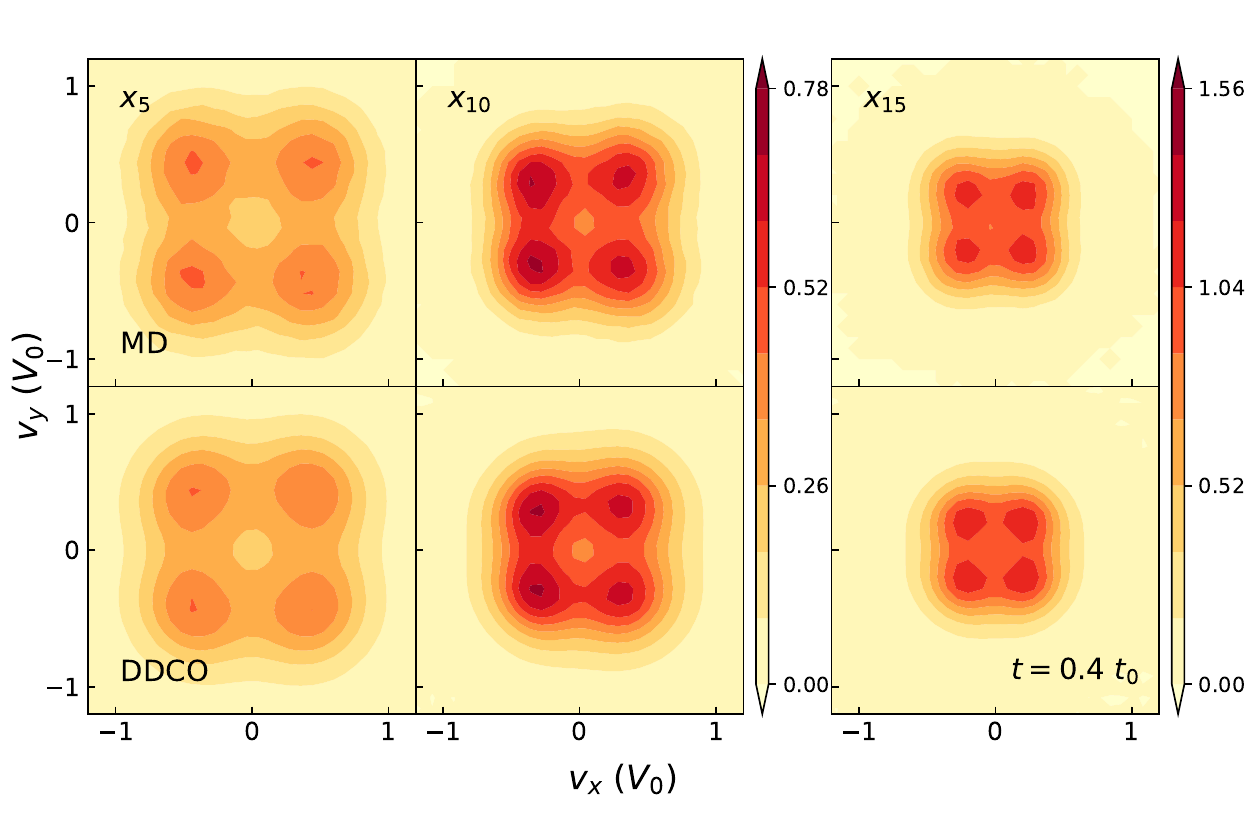} \\
    \includegraphics[width=0.7\textwidth]{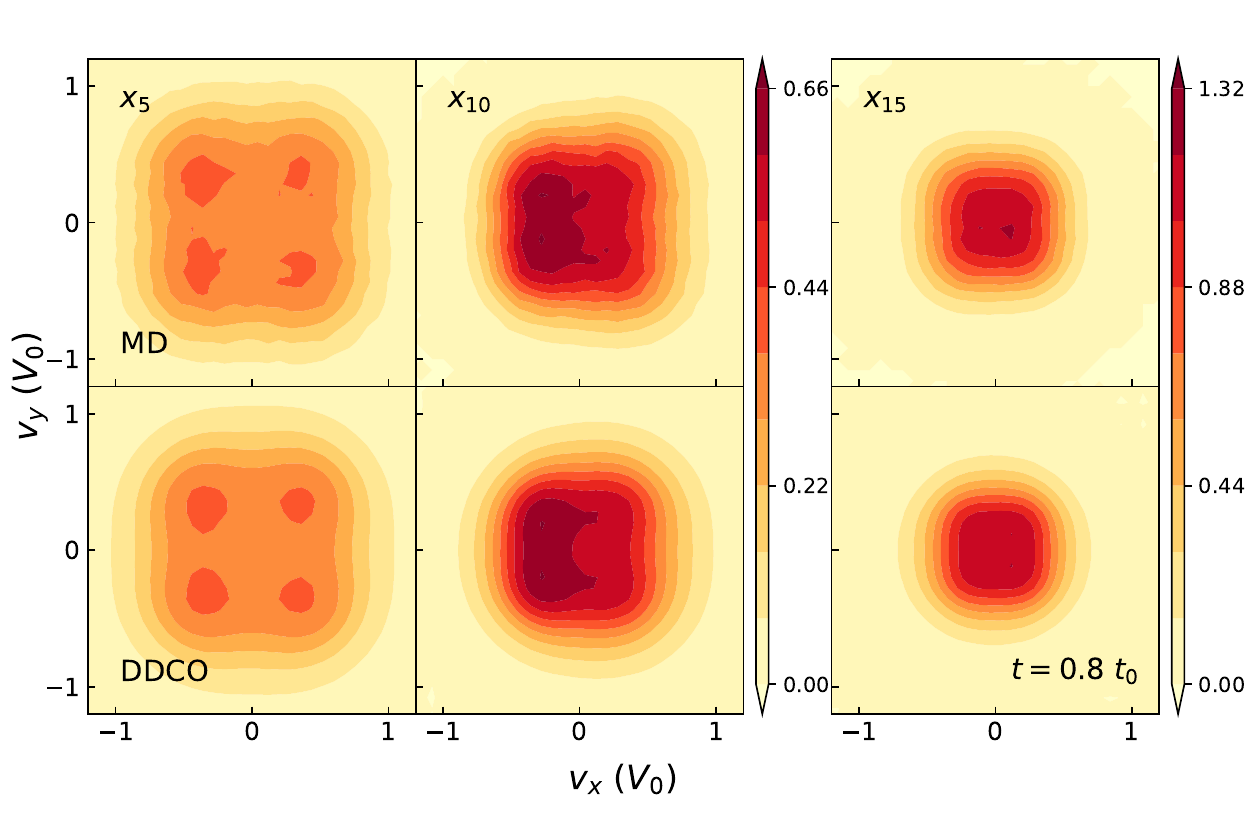}
    \caption{Comparison of the instantaneous velocity distribution on the $v_x$-$v_y$ plane at $t=0.4~t_{0}$ (up) and $0.8~t_{0}$ (bottom) with the initial symmetric double-well distribution. ``DDCO'' represents our model with the data-driven collision operator.}
    % ``DD-Landau'' represents the simplified version of data-driven Landau model
    \label{fig:slice_dw1}
\end{figure}

Fig. \ref{fig:slice_dw1} and \ref{fig:slice_dw2} present the instantaneous velocity distribution function $f(t, \bm{v}; x)$ on the $v_x$-$v_y$ plane at different $x$ values for the symmetric double-well and asymmetric double-well as the initial conditions. For both $t=0.4~t_{0}$ and $t=0.8~t_{0}$, the prediction of the present DDCO model shows good agreement with the full MD results. Specifically, $f(\bm v; x)$ at $x=15\Delta x$ (i.e., the lowest temperature) shows the fast relaxation to equilibrium while $f(\bm v; x)$ at $x=5\Delta x$ (i.e., the highest temperature) shows the slowest relaxation as expected.

\begin{figure}[H]
    \centering
    \includegraphics[width=0.7\textwidth]{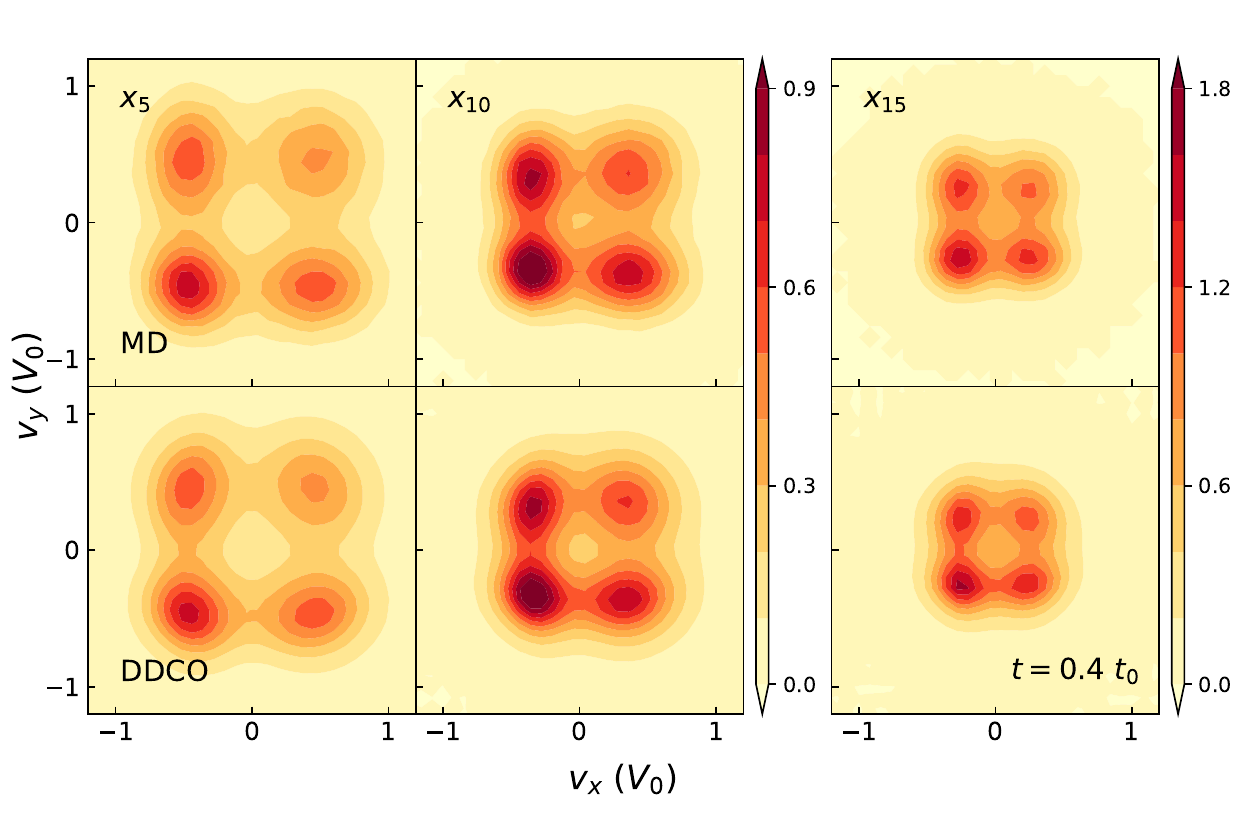} \\
    \includegraphics[width=0.7\textwidth]{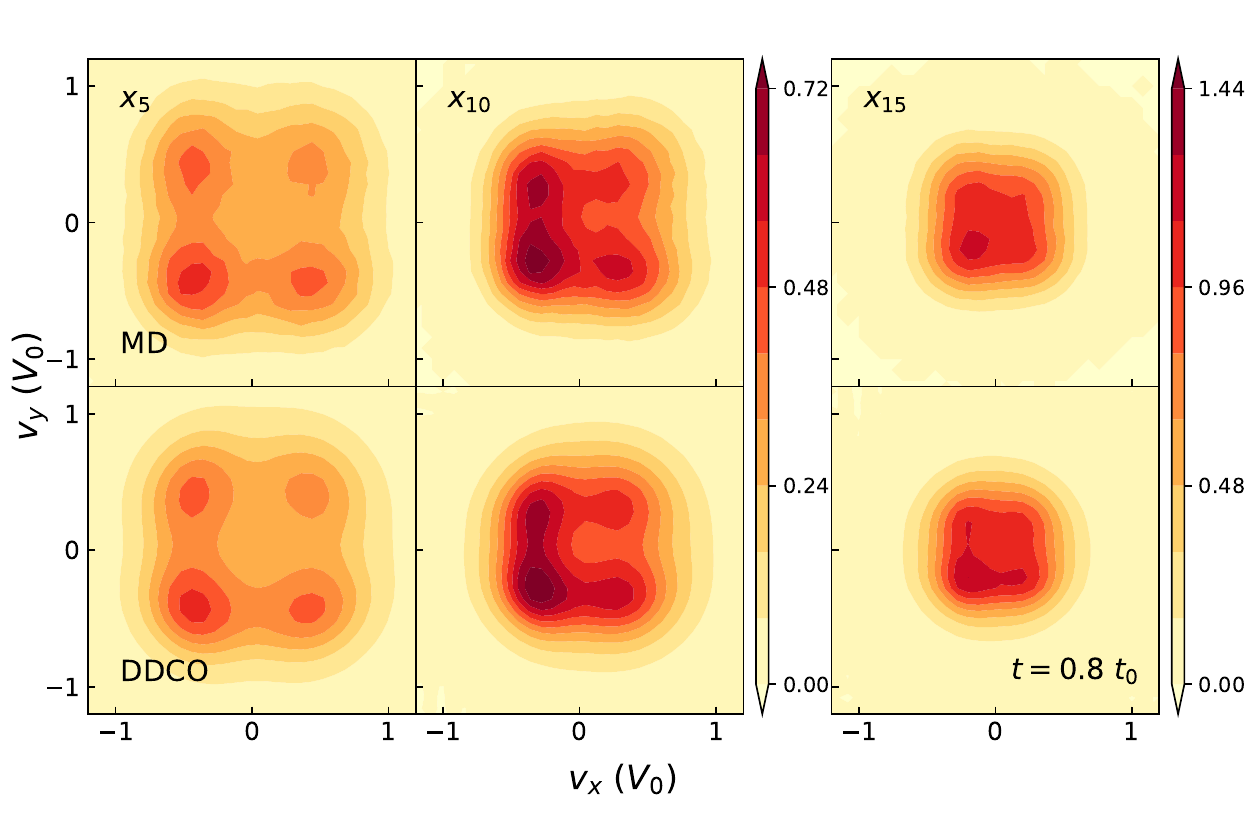}
    \caption{Comparison of the instantaneous velocity distribution on the $v_x$-$v_y$ plane at $t=0.4~t_{0}$ (up) and $0.8~t_{0}$ (bottom) with the initial asymmetric double-well distribution.}
    \label{fig:slice_dw2}
\end{figure}

Fig. \ref{fig:xvx_dw1} and \ref{fig:xvx_dw2} further show the prediction of  $f(t, \bm{v}; x)$ on the $x$-$v_x$ plane at $t=0.4~t_{0}$ and $0.6~t_{0}$. The predictions of the present DDCO model show good agreement with the MD results over the full spatial region. Such consistency validates that the prediction model can accurately predict the plasma kinetics over a broad range of plasma physical conditions. 

% In contrast, the prediction from the Landau model fails to predict.

% We also compare the velocity distribution on the $x$-$v_{x}$ plane at $t=0.4~t_{0}$ and $0.6~t_{0}$ in Fig. \ref{fig:xvx_dw1}.

\begin{figure}[H]
    \centering
    \includegraphics[width=0.48\textwidth]{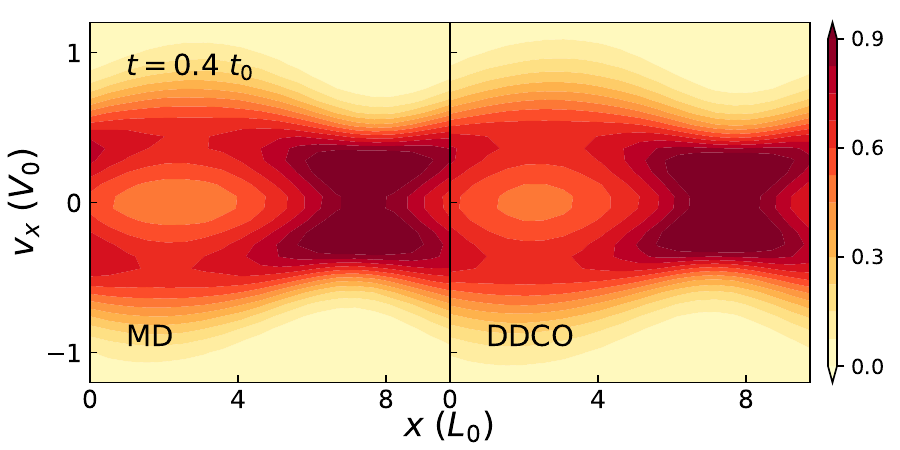}
    \includegraphics[width=0.48\textwidth]{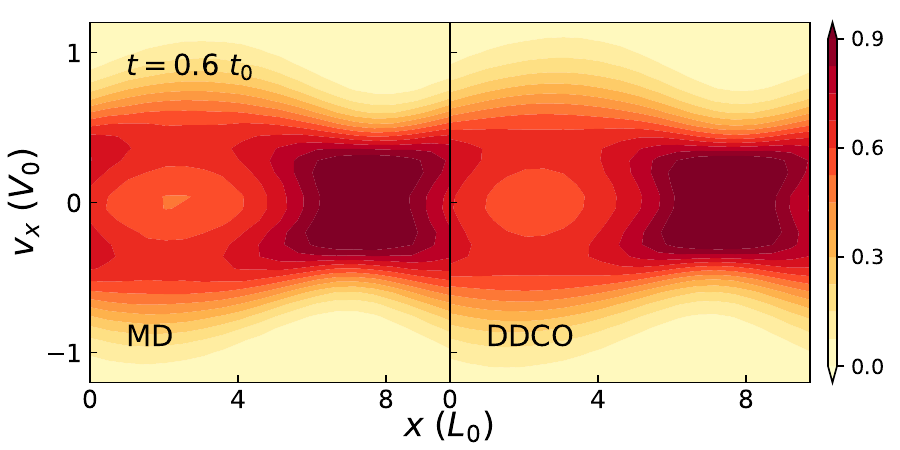}
    \caption{Comparison of the instantaneous velocity distribution on the $x$-$v_{x}$ plane at $t=0.4~t_{0}$ (left) and $0.6~t_{0}$ (right) with the initial symmetric double-well distribution.}
    \label{fig:xvx_dw1}
\end{figure}

\begin{figure}[H]
    \centering
    \includegraphics[width=0.48\textwidth]{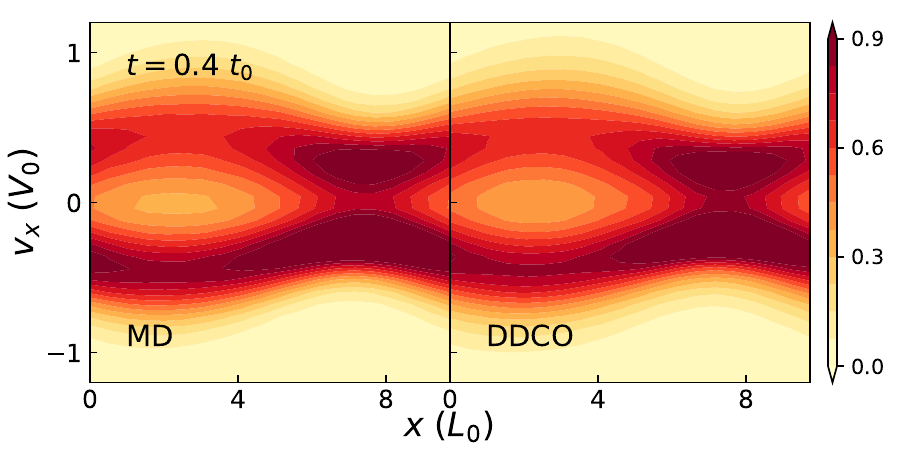}
    \includegraphics[width=0.48\textwidth]{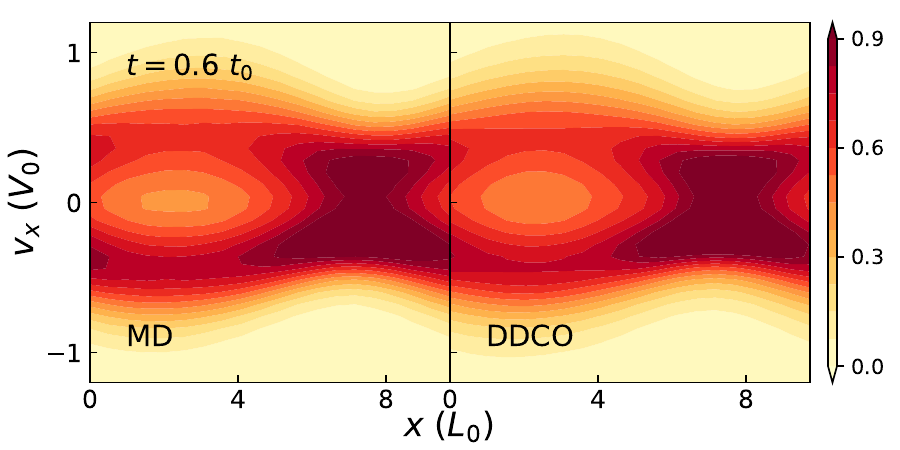}
    \caption{Comparison of the instantaneous velocity distribution on the $x$-$v_{x}$ plane at $t=0.4~t_{0}$ (left) and $0.6~t_{0}$ (right) with the initial asymmetric double-well distribution.}
    \label{fig:xvx_dw2}
\end{figure}

%The computational cost of our model is only several-fold of the Landau equation.
Fig. \ref{fig:MPE} illustrates the conservation of total mass and total energy along the numerical simulation for the initial local Maxwellian and bi-Maxwellian distributions. The total momentum exhibits some oscillations only along the x direction.
This demonstrates that our second-order numerical scheme guarantees the conservation of physical quantities and therefore is applicable for long-time simulation. 

\begin{figure}[H]
    \centering
    \includegraphics[width=0.48\textwidth]{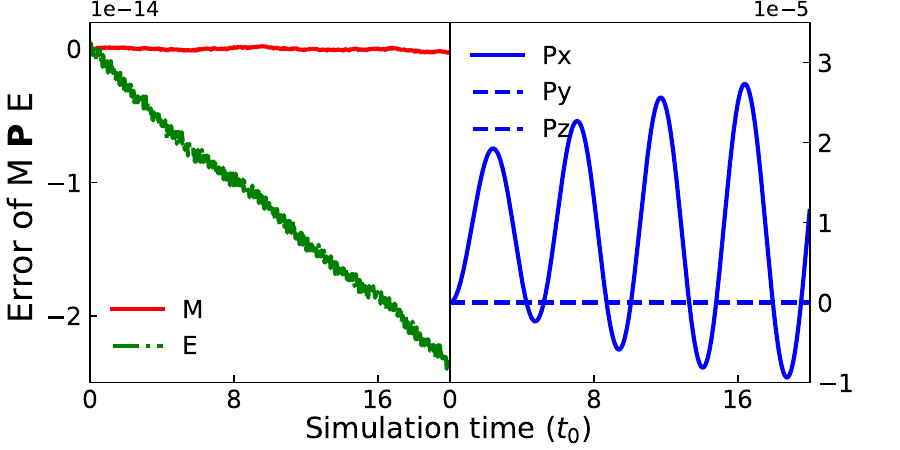} 
    \includegraphics[width=0.48\textwidth]{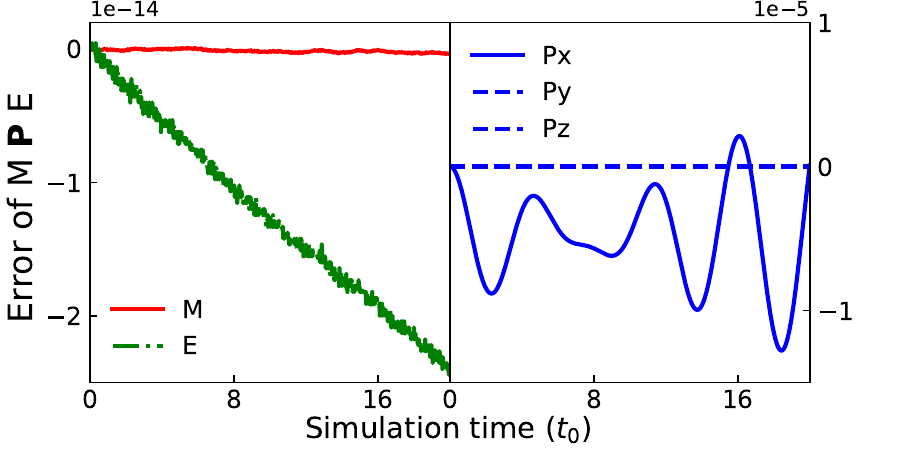}
    \caption{The error of total mass $M$, momentum $\boldsymbol{P}$, and kinetic energy $E$ along with the simulation time, with the initial local Maxwellian (left) and bi-Maxwellian (right) distributions.}
    \label{fig:MPE}
\end{figure}

As the present study focuses on the structure-preserving construction for 1D-3V collision operator, we use an explicit numerical scheme to illustrate the essential idea. Accordingly, we are not able to prove that the advection term preserves entropy for the present full discrete scheme. Since the Vlasov component does not contribute to the evolution of entropy in the continuous sense, a Taylor expansion can be employed to estimate that the change in entropy associated with the Vlasov component is approximately $\cO(\Delta t^2)$ at each time step.
Therefore, the change in entropy at each step is dominated by the collision term, which is approximately $\cO(\Delta t)$.
The numerical results in Fig. \ref{fig:etp} show that the discrete entropy consistently increases monotonically for various initial conditions.
\begin{figure}[H]
    \centering
    \includegraphics[width=0.36\textwidth]{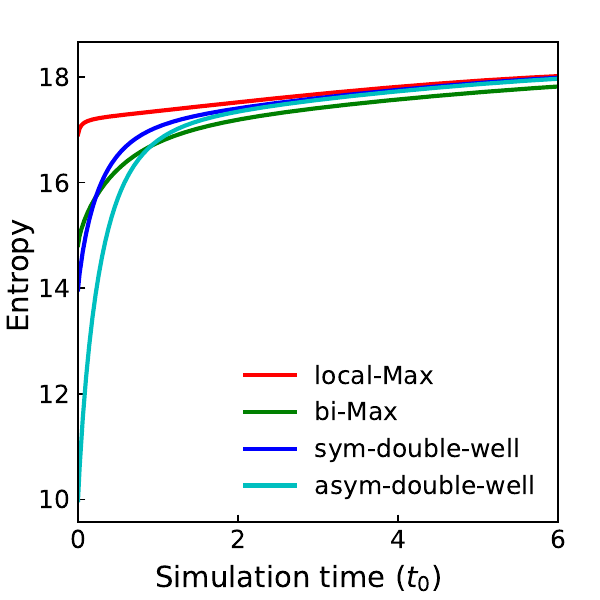}
    \caption{The entropy evolution of 1D3V plasma systems with different initial velocity distributions, including local-Maxwellian, bi-Maxwellian, symmetric double-well and asymmetric double-well distributions.}
    \label{fig:etp}
\end{figure}

\section{Summary}\label{sec:summary}

This work introduces a generalized structure-preserving data-driven collision operator into the Vlasov-Amp\`ere system for 1D-3V kinetic model. The collision operator is constructed based on the metriplectic framework and can be applied to 
spatially inhomogeneous plasma systems over a broad range of physical conditions. 
%This across the weak and moderate coupling regimes. 
The operator incorporates an anisotropic, non-stationary collision kernel, directly learned from MD simulations, and explicitly depends on local plasma density and temperature, enabling accurate descriptions of collisional processes across the weak and moderate coupling regimes, where the classical Landau operator fails.

The constructed operator strictly preserves the conservation laws, symmetry constraints, and the H-theorem. Computation efficiency is achieved by a low-rank tensor representation of the generalized collision kernel with $\cO(N\log N)$ computational complexity. We develop an explicit, second-order, energy-conserving numerical scheme to conserve fully discrete mass and total energy. Numerical results demonstrate that the proposed model accurately predict both the transport coefficients and the plasma kinetics in comparison with the full MD simulations across a wide range of densities and temperatures in 1D3V settings. 

The present framework bridges the micro-scale first principle descriptions encoded with the heterogeneous many-body effects with the meso-scale kinetic models beyond the empirical form while strictly preserve the essential physical properties on the PDE perspective. For illustration purpose, the numerical scheme is based on a standard second-order explicit time discretization, where the numerical entropy production can not be theoretically guaranteed. Future work will be devoted to develop more accurate numerical discretization scheme with theoretical guarantee for the conservation laws and H-theorem on the discrete level, as well as the applications to other complex plasma kinetic processes such as the multi-species systems.  

% non-uniform grids
% We find that, in the 1D3V example, when the temperature changes significantly, a larger region (due to high temperature) and a smaller finite difference scale (due to low temperature) are required, resulting in an excessive number of discrete points and excessive memory overhead.
% In future work, we aim to extend this formulation to fixed velocity regions and finite difference scale (that do not vary with temperature), making it applicable to plasma collisions and relaxation processes at different temperatures, as well as collisions of two or more plasmas with large mass ratios.

\appendix
%\section{Appendix}\label{sec:appendix}

\section{Momentum conservation numerical scheme}
\label{app:mom_conserv_scheme}

The Vlasov–Amp\`ere–collision (VAC) system is equivalent to the Vlasov–Poisson–collision (VPC) system when the charge continuity equation $\rho_{t} + \nabla_{\bm{x}} \cdot \bm{J} = 0$ holds, while the Amp\`ere solver only conserves mass and energy, rather than the momentum.
The VPC is written as
% We only consider the Vlasov equation here, since the collision naturally conserves mass, momentum and energy.
\begin{equation}\label{eq:VP}
\begin{aligned}
    \partial_{t} f + v_{x} \partial_{x} f + E(x,t)\,\nabla_{v_{x}} f = C[f], \qquad (x,\bm{v}) \in \Omega_{x} \times \Omega_{\bm{v}}, \\
    -\partial_{xx}\phi = \rho' := \lambda_{D}^{-2}(\rho - \bar{\rho}), \qquad \rho(x,t) = \int_{\Omega_{\bm{v}}} f(x,\bm{v},t)\,\mathrm{d}^{3}\bm{v} .
\end{aligned}
\end{equation}

We propose the fully discrete Vlasov–Poisson solver to conserve mass and momentum with one-order numerical error on total energy.
The discretization is 
\begin{equation}
\begin{aligned}
    & \dfrac{f_{i,\bm{j}}^{n+1} - f_{i,\bm{j}}^{n}}{\Delta t} + v_{j_{x}} D_{x}^{up} f_{i,\bm{j}}^{n} + E_{i}^{n} D_{v}^{cen} f_{i,\bm{j}}^{n} = 0 , \\
    & D_{x}^{up} f_{i,\bm{j}}^{n} = \begin{cases}
        \dfrac{f_{i,\bm{j}}^{n} - f_{i-1,\bm{j}}^{n}}{\Delta x}, & v_{\bm{j}} > 0 , \\
        \dfrac{f_{i+1,\bm{j}}^{n} - f_{i,\bm{j}}^{n}}{\Delta x}, & v_{\bm{j}} < 0 ,
    \end{cases} \\
    & D_{v}^{cen} f_{i,\bm{j}}^{n} = \dfrac{f_{i,\bm{j}+1_{x}}^{n} - f_{i,\bm{j}-1_{x}}^{n}}{2 \Delta v} ,
\end{aligned}
\end{equation}
where we adopt the upwind scheme in the spatial advection and the central flux scheme in the velocity advection term.

\begin{proposition}
    This discretization of the Vlasov-Poisson equation preserves the discrete mass and momentum, while and kinetic energy s not conserved.
    % and satisfies a discrete entropy inequality under the CFL condition $\frac{\Delta t}{\Delta x}\max |v_x|\le 1$.
\end{proposition}
\begin{proof}
The discrete mass, momentum and energy are defined in Eq. \eqref{eq:MPE},
% \begin{equation}
%     \begin{aligned}
%         M^{n} =& \Delta x \Delta v \sum_{i,\bm{j}} f_{i,\bm{j}}^{n}, ~~ &\bm{P}^{n} =& \Delta x \Delta v \sum_{i,\bm{j}} \bm{v}_{\bm{j}} f_{i,\bm{j}}^{n}, \\
%         \mathcal{E}_{K}^{n} =& \Delta x \Delta v \sum_{i,\bm{j}} \dfrac{\bm{v}_{\bm{j}}^{2}}{2} f_{i,\bm{j}}^{n}, ~~ &\mathcal{E}_{P}^{n} =& \Delta x \sum_{i} \dfrac{\lambda_{D}^{2}}{2} |E_{i}^{n}|^{2} , \\
%         \mathcal{E}^{n} =& \mathcal{E}_{K}^{n} + \mathcal{E}_{P}^{n}, ~~ &S^{n} =& - \Delta x \Delta v \sum_{i,\bm{j}} f_{i,\bm{j}}^{n} \log f_{i,\bm{j}}^{n} .
%     \end{aligned}
% \end{equation}

\paragraph{Mass conservation}

\begin{equation}
    \dfrac{M^{n+1}-M^n}{\Delta t} = - \Delta x \Delta v \sum_{\bm{j}} v_{j_{x}} \sum_{i} D_{x}^{up} f_{i,\bm{j}}^{n} - \Delta x \Delta v \sum_{i} E_{i}^{n} \sum_{\bm{j}} D_{v}^{cen} f_{i,\bm{j}}^{n} = 0 ,
\end{equation}
where the first part vanishes due to telescoping cancellation, and the second part vanishes due to zero boundary condition.
% $f_{i,\bm{j}}^{n} = 0$ at boundary.

\paragraph{Momentum conservation}

\begin{equation}
\begin{aligned}
    \dfrac{\bm P^{n+1}-\bm P^n}{\Delta t} &= - \Delta x \Delta v \sum_{\bm{j}} v_{\bm{j}} v_{j_{x}} \sum_{i} D_{x}^{up} f_{i,\bm{j}}^{n} - \Delta x \Delta v \sum_{i} E_{i}^{n} \sum_{\bm{j}} v_{\bm{j}} D_{v}^{cen} f_{i,\bm{j}}^{n} \\
    &= \Delta x \sum_{i} E_{i}^{n} \rho_{i}^{n} = 0 ,
\end{aligned}
\end{equation}
where $\rho_{i}^{n} = \Delta v \sum_{\bm{j}} f_{i,\bm{j}}^{n}$.
The first term vanishes due to telescoping cancellation, and the second equation vanishes when we use FFT or central difference scheme on the Poisson equation.

\paragraph{Energy change}

\begin{equation}
    \dfrac{\mathcal E_K^{n+1}-\mathcal E_K^n}{\Delta t} = \Delta x \sum_{i} E_{i}^{n} J_{i}^{n},
\end{equation}
where $J_{i}^{n} = \Delta v \sum_{\bm{j}} \bm{v}_{\bm{j}} f_{i,\bm{j}}^{n}$, and it is similar in the momentum derivation.

The potential energy change is 
\begin{equation}
    \dfrac{\mathcal{E}_{P}^{n+1}-\mathcal{E}_{P}^n}{\Delta t} = \dfrac{\lambda_{D}^{2} \Delta x}{2 \Delta t} \sum_{i} (E_{i}^{n+1} + E_{i}^{n}) (E_{i}^{n+1} - E_{i}^{n}) ,
\end{equation}
and $E_{i}^{n+1} - E_{i}^{n} = - \lambda_{D}^{-2} J_{i}^{n} \Delta t + \cO(\Delta x \Delta t)$,
so that 
\begin{equation}
\begin{aligned}
    \dfrac{\mathcal{E}_{P}^{n+1}-\mathcal{E}_{P}^{n}}{\Delta t} =& \dfrac{\lambda_{D}^{2} \Delta x}{2 \Delta t} \sum_{i} [2E_{i}^{n} - \lambda_{D}^{-2} J_{i}^{n} \Delta t + \cO(\Delta x \Delta t)] [- \lambda_{D}^{-2} J_{i}^{n} \Delta t + \cO(\Delta x \Delta t)] \\
    =& - \Delta x \sum_{i} E_{i}^{n} J_{i}^{n} + \sum_{i} \dfrac{\lambda_{D}^{-2} \Delta t \Delta x}{2} |J_{i}^{n}|^{2} + \cO(\Delta x) + \cO(\Delta x \Delta t) .
\end{aligned}
\end{equation}

Therefore, there is a first-order time error in the total energy $\frac{\mathcal{E}_{P}^{n+1}-\mathcal{E}_{P}^{n}}{\Delta t} = \cO(\Delta t)$.
% because in the central-difference Poisson
% \begin{equation*}
% \begin{aligned}
%     (\rho_{i}^{n+1} - \rho_{i}^{n}) / \Delta t = - (J_{i}^{n} - J_{i-1}^{n}) / \Delta x + \cO(\Delta x) = - D_{x} J_{i}^{n} + \cO(\Delta x), \\
%     D_{x} (E_{i}^{n+1} - E_{i}^{n}) = \lambda_{D}^{-2} (\rho_{i}^{n+1} + \rho_{i-1}^{n+1} - \rho_{i}^{n} - \rho_{i-1}^{n})/2
% \end{aligned}
% \end{equation*}

% \paragraph{Entropy production.}
% Under the CFL condition $\frac{\Delta t}{\Delta x}|v_{j_x}|\le 1$, the update for each fixed velocity index $\bm j$ can be written as a convex combination
% \begin{equation*}
%     f_{i,\bm j}^{n+1} = \alpha_{i-1} f_{i-1,\bm j}^n + \alpha_i f_{i,\bm j}^n + \alpha_{i+1} f_{i+1,\bm j}^n, \quad \alpha_k\ge 0,\; \sum_k \alpha_k=1.
% \end{equation*}
% For any concave function $\eta$, Jensen's inequality implies
% \begin{equation*}
%     \eta(f_{i,\bm j}^{n+1}) \geq \alpha_{i-1}\eta(f_{i-1,\bm j}^n) + \alpha_i\eta(f_{i,\bm j}^n) + \alpha_{i+1}\eta(f_{i+1,\bm j}^n).
% \end{equation*}
% Summing over $i$ and $\bm j$ and choosing $\eta(f)=-f\log f$ gives the discrete entropy production.
% \end{proof}

\end{proof}

\section{Properties and construction of the inhomogeneous collision kernel}
\label{app:prop}

The collision operator in Eq. \eqref{eq:collision_metripletic} is
\begin{equation*}
    \begin{aligned}
        C[f] = \nabla_{\bm{v}} \cdot \int \bm{\omega}(\bm{v}, \bm{v}'; \rho, T) \left[ f(\bm{v}') \nabla f(\bm{v}) - f(\bm{v}) \nabla' f(\bm{v}') \right] \mathrm{d}\bm{v}', \\
        \rho = \int f \mathrm{d}^{3} \bm{v}, ~~ \bar{\bm{v}} = \int \bm{v} f ~ \mathrm{d}^{3} \bm{v} / \rho, ~~ T = \int \dfrac{(\bm{v} - \bar{\bm{v}})^{2}}{2} f ~ \mathrm{d}^{3} \bm{v} / \rho ,
    \end{aligned}
\end{equation*}
where the collision kernel $\bm{\omega}(\bm{v}, \bm{v}'; \rho, T)$ needs to be symmetric positive semi-definite and satisfies the conditions of variable exchange symmetry, rotational symmetry, and zero projection for all cases of $(\rho, T)$:
\begin{equation}\label{eq:kernel_conditions}
    \begin{aligned}
        \bm{\omega}(\mathcal{U}\bm{v},\mathcal{U}\bm{v}') &= \mathcal{U}\bm{\omega}(\bm{v},\bm{v}')\mathcal{U}^T,~  \mathcal{U} \in {\rm SO}(3) , \\
        \bm{\omega}(\bm{v},\bm{v}') &= \bm{\omega}(\bm{v}',\bm{v}) , \\
        \bm{\omega}(\bm{v},\bm{v}')(\bm{v} -\bm{v}') &= \bm{0} , \\
        \bm{\omega}(\bm{v},\bm{v}') &= \bm{\omega}(-\bm{v},-\bm{v}') ,
    \end{aligned}
\end{equation}
such that the collision operator satisfies the conservation laws \eqref{eq:conservation_law}, H-theorem, and frame-indifference constraints \eqref{eq:frame_indifference}.
Then we employ a low-rank tensor representation of the anisotropic non-stationary kernel in \eqref{eq:SS} as 
\begin{equation}\label{eq:CM2}
\begin{aligned}
    \bm{\omega} (\bm{v}, \bm{v}'; \rho, T) =& g_{1}^{2} |\bm{\mathcal{P}}\bm{r}|^{2} \bm{\mathcal{P}} + (g_{2}^{2}-g_{1}^{2}) \bm{\mathcal{P}} \bm{r} \bm{r}^{T} \bm{\mathcal{P}} , \\
    g_{\ast}(\bm{v}, \bm{v}'; \rho, T) =& \sum_{j=1}^{J} L_{\ast}^{j}(|\bm{u}|; \rho, T) M_{\ast}^{j}(|\bm{v}|; \rho, T) N_{\ast}^{j}(|\bm{v}'|; \rho, T) ,
\end{aligned}
\end{equation}
which enables a convolution structure in the simulation, achieving $\cO(N \log(N))$ computational cost.
% the integral term preserves a convolution structure
% enables us to achieve efficient evaluation of the collision operator.
For example the first term in \eqref{eq:sim_C} is
\begin{equation}\label{eq:sim_I1}
    \begin{aligned}
        I_{1} =& \nabla \cdot \int g_{1}^{2} |\bm{\mathcal{P}}\bm{r}|^{2} \bm{\mathcal{P}} f(\bm{v}) f(\bm{v}') \left[ \nabla \log f(\bm{v}) - \nabla'\log f(\bm{v}') \right] \mathrm{d}\bm{v}' \\
        =& \sum_{j,k} \nabla \cdot 
        \left[ \{2|\bm{v}|^2,\bm{\mathcal{P}},1\}_{1}^{j,k} 
        +\{2,\bm{\mathcal{P}},|\bm{v}'|^{2}\}_{1}^{j,k} 
        +\{-1,|\bm{u}|^{2}\bm{\mathcal{P}},1\}_{1}^{j,k} \right. \\
        &\qquad\quad \left. 
        +\{-|\bm{v}|^{4},|\bm{u}|^{-2}\bm{\mathcal{P}},1\}_{1}^{j,k} 
        +\{ 2|\bm{v}|^{2},|\bm{u}|^{-2}\bm{\mathcal{P}},|\bm{v}'|^{2}\}_{1}^{j,k} 
        +\{-1,|\bm{u}|^{-2}\bm{\mathcal{P}},|\bm{v}'|^{4}\}_{1}^{j,k} \right] ,
    \end{aligned}
\end{equation}
where individual terms $\{ \bm{\alpha}(\bm{v}),\bm{\beta}(\bm{u}),\bm{\gamma}(\bm{v}')\}_{\ast}^{j,k}$ take the form
\begin{equation}
    \begin{aligned}
        \{ \bm{\alpha}(\bm{v}),\bm{\beta}(\bm{u}),\bm{\gamma}(\bm{v}')\}_{\ast}^{j,k} =& M_{\ast}^{j}(|\bm{v}|) M_{\ast}^{k}(|\bm{v}|) \bm{\alpha}(\bm{v}) \cdot \left[ \langle\bm{\beta}(\bm{u}), \bm{\gamma}(\bm{v}')f(\bm{v}') \rangle_{\ast}^{j,k} f(\bm{v}) \nabla \log f(\bm{v}) \right.\\
        & \qquad \qquad \qquad \qquad \qquad - \left. \langle \bm{\beta}(\bm{u}), \bm{\gamma}(\bm{v}') f(\bm{v}') \nabla' \log f(\bm{v}') \rangle_{\ast}^{j,k} f(\bm{v}) \right],
    \end{aligned}
\nonumber
\end{equation}
and $\langle \bm{\xi}(\bm{u}),\bm{\eta}(\bm{v}') \rangle_{\ast}^{j,k} $ represents the convolution structure, i.e.,  
\begin{equation*}
\langle \bm{\xi}(\bm{u}),\bm{\eta}(\bm{v}') \rangle_{\ast}^{j,k} = \int  L_{\ast}^{j}(|\bm{u}|) L_{\ast}^{k}(|\bm{u}|) \bm{\xi}(\bm{u}) \cdot N_{\ast}^{j}(|\bm{v}'|) N_{\ast}^{k}(|\bm{v}'|) \bm{\eta}(\bm{v}') \mathrm{d}\bm{v}',    
\end{equation*} 

% The integral term $I_2$ can be expressed as
% \begin{equation}\label{eq:I2}
%     \begin{aligned}
%         I_{2} &= \nabla \cdot \int g_{2}^{2} \bm{\mathcal{P}} \bm{r} \bm{r}^{T} \bm{\mathcal{P}} f(\bm{v}) f(\bm{v}') \left[ \nabla \log f(\bm{v}) - \nabla'\log f(\bm{v}') \right] \mathrm{d}\bm{v}' \\
%         &= \sum_{j,k} \nabla \cdot 
%         \left[ \{2\bm{v},\bm{u}^{T},1\}_{2}^{j,k} 
%         + \{4\bm{v},1,\bm{v}'^{T}\}_{2}^{j,k} 
%         + \{-1,\bm{u}\bm{u}^{T},1\}_{2}^{j,k} \right. \\
%         &\qquad\qquad \left. + \{-2,\bm{u},\bm{v}'^{T}\}_{2}^{j,k} 
%         + \{-2|\bm{v}|^{2}\bm{v},|\bm{u}|^{-2}\bm{u}^{T},1\}_{2}^{j,k} 
%         + \{-2|\bm{v}|^{2},|\bm{u}|^{-2}\bm{u},\bm{v}'^{T}\}_{2}^{j,k} \right. \\
%         &\qquad\qquad \left. 
%         + \{2\bm{v},|\bm{u}|^{-2}\bm{u}^{T},|\bm{v}'|^{2}\}_{2}^{j,k} 
%         + \{2,|\bm{u}|^{-2}\bm{u},|\bm{v}'|^{2}\bm{v}'^{T}\}_{2}^{j,k} 
%         + \{|\bm{v}|^{4},|\bm{u}|^{-4}\bm{u}\bm{u}^{T},1\}_{2}^{j,k} \right. \\
%         &\qquad\qquad \left. + \{-2|\bm{v}|^{2},|\bm{u}|^{-4}\bm{u}\bm{u}^{T},|\bm{v}'|^{2}\}_{2}^{j,k} 
%         + \{1,|\bm{u}|^{-4}\bm{u}\bm{u}^{T},|\bm{v}'|^{4}\}_{2}^{j,k} \right] , 
%     \end{aligned}
% \end{equation}

\section{Comparison with homogeneous kernel} % Dimensional analysis and characteristic parameters
\label{app:comp_homo}

The present generalized data-driven collision operator admits a non-stationary, symmetry-breaking form that captures the heterogeneous energy transfer arising from the many-body effects. This effect becomes crucial for predicting the plasma kinetics beyond the weakly coupled regime. We refer to Ref. \cite{zhao2025data} for detailed discussion. 

For illustration purpose, we further propose a simplified data-driven homogeneous collision kernel, named ``DD-Landau'', with the form
\begin{equation}
     \bm{\omega} (\bm{v}, \bm{v}'; \rho, T) = g_{0}^{2}(\rho,T) \boldsymbol{\mathcal{P}} / |\bm{u}| ,
\label{eq:DD_Landau}     
\end{equation}
where the scalar function $g_{0}(\rho,T)$ is constructed as a neural network, and trained in a similar way as DDCO. 
The ``DD-Landau'' model has a similar form to the traditional Landau equation, while the traditional Landau equation is invalid with a negative factor $\ln \Lambda$ with $\Lambda \sim \Gamma^{-3/2}$ for the present plasma condition $\Gamma > 1$.  

Figs. \ref{fig:app_slice_dw1} show that the predictions of the DD-Landau and the present DDCO in comparison with the full MD results. The prediction of the DD-Landau shows apparent discrepancy with the MD results. In contrast, the prediction of the present DDCO model shows good agreement. This result further verifies that the non-stationary, symmetry-breaking form of the DDCO model (see Eqs. \eqref{eq:kernel_form} and \eqref{eq:SS}) is crucial for encoding the unresolved many-body correlations for moderately coupled regime. 

\begin{figure}[H]
    \centering
    \includegraphics[width=0.9\textwidth]{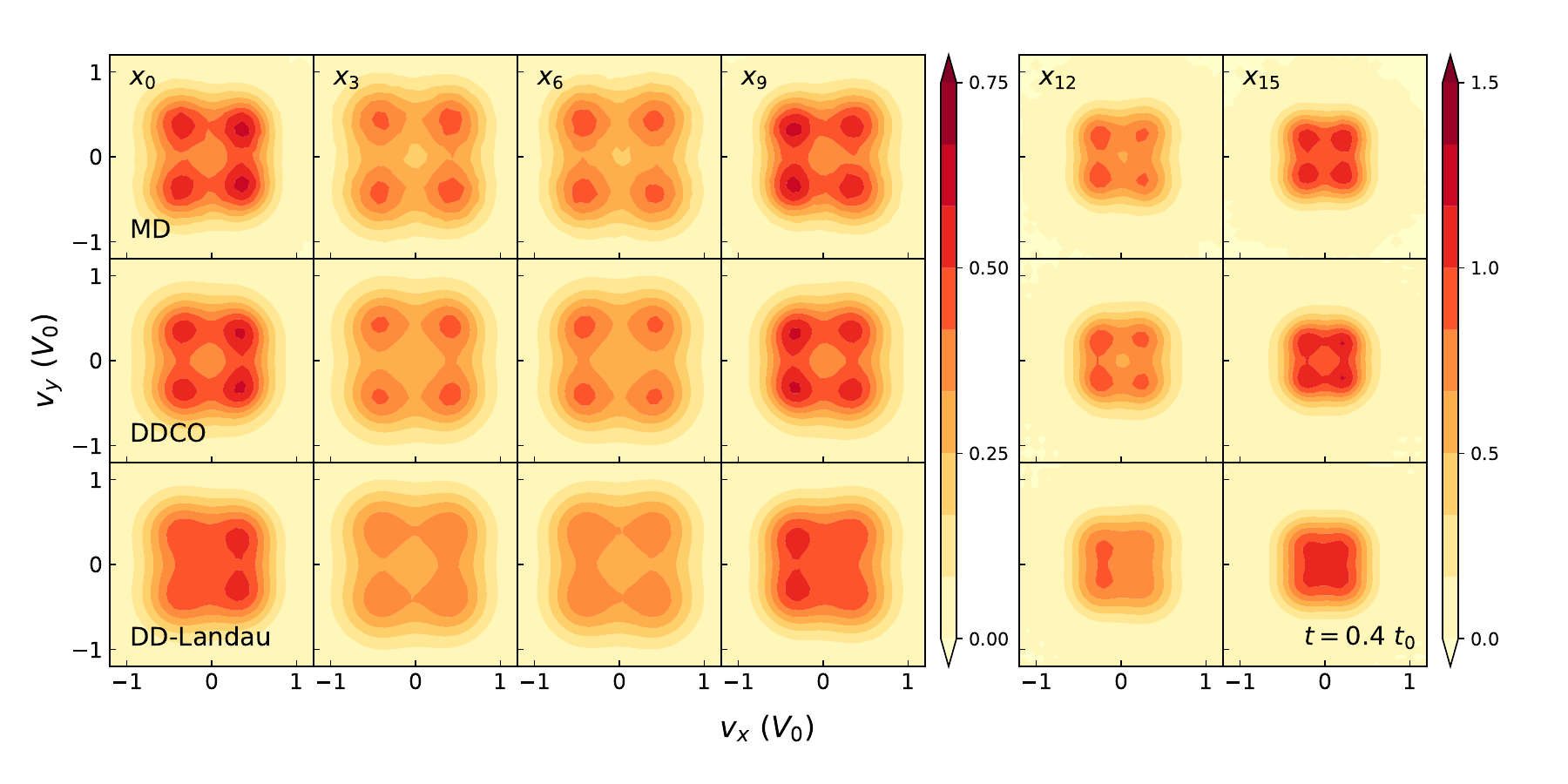}
    \caption{Comparison of the instantaneous velocity distribution on the $v_x$-$v_y$ plane at $t=0.4~t_{0}$ with the initial symmetric double-well distribution. The DD-Landau takes a homogeneous kernel in Eq. \eqref{eq:DD_Landau} similar to the standard Landau model, where the pre-factor $g_0^2(\rho, T)$ is directly trained from the micro-scale MD. The discrepancy between the DD-Landau and the full MD indicates the crucial role of the symmetry-breaking form of the present DDCO model.}
    \label{fig:app_slice_dw1}
\end{figure}

% \begin{figure}[H]
%     \centering
%     \includegraphics[width=0.9\textwidth]{./figs/ver5_4_5/slice_3.pdf}
%     \caption{Comparison of the instantaneous velocity distribution on the $v_x$-$v_y$ plane at $t=0.6~t_{0}$ with the initial asymmetric double-well distribution.}
%     \label{fig:app_slice_dw2}
% \end{figure}

% \begin{figure}[H]
%     \centering
%     \includegraphics[width=0.9\textwidth]{./figs/ver5_4_6/slice_3.pdf}
%     \caption{Comparison of the instantaneous velocity distribution on the $v_x$-$v_y$ plane at $t=0.6~t_{0}$ with the initial timodal distribution.}
%     \label{fig:app_slice_tri}
% \end{figure}

\section{Physical characteristic and settings} 
% Physical characteristic quantity analysis
% Dimensional analysis and characteristic parameters
\label{app:phys}

The Vlasov-Amp\`ere equation can be applied to describe OCP systems with Knudsen number $\Kn = \lambda / L_{0} = \cO(1)$, characterized by the Coulomb collision frequency $\nu \propto \frac{n_{0} q_{0}^{4} \ln \Lambda}{(4\pi \epsilon_{0})^{2} m_{0}^{2} v_{th}^{3}}$, with Coulomb logarithm $\ln \Lambda$, thermal velocity $v_{th} = \sqrt{k_{B} T/m}$, characteristic length $L_{0}$, and mean free path $\lambda \sim v_{th} / \nu$.
So that $\Kn \propto T^{2}/n_{0} L_{0}$. % for plasma systems
% mean free path $\lambda = 1/(\sqrt{2}\pi d^{2}n)$ and characteristic length $L_{0}$ for neutral dilute gas systems.

% further depends on the local particle density and temperature.
The prefactor in the Landau model $\gamma \propto \ln\Lambda$ depends on the Coulomb logarithm $\ln\Lambda = \ln (\lambda_{D} / b_{min})$, with $\lambda_{D} = \sqrt{\frac{\epsilon_{0} k_{B} T}{n_{0} q_{0}^{2}}}$ and $b_{min} = \frac{q_{0}^{2}}{4\pi \epsilon_{0} k_{B} T}$.
For a classical OCP, the Coulomb parameter scales as $\Lambda \sim \Gamma^{-3/2}$, where $\Gamma = \frac{q_e^2}{4\pi \epsilon_0 k_B T}(4\pi n/3)^{1/3}$, implying that weak coupling (the coupling parameter $\Gamma \ll 1$) corresponds to large Coulomb logarithms and justifies the Landau approximation, while moderate or strong coupling ($\Gamma = \cO(1)$) lies beyond its range of validity.
% The coupling parameter $\Gamma = \frac{q_e^2}{4\pi \epsilon_0 k_B T}(4\pi n/3)^{1/3}$.

% and numerical parameters used in simulations
In the numerical simulations, we choose $n_{0}=10^{24} ~ m^{-3}$, $q_{0} = 1.6 \times 10^{-19} ~ C$, $m_{0} = 1.67 \times 10^{-27} ~ kg$, $k_{B} T = 0.1 \sim 0.3 ~ eV$, and the characteristic length $L_{0} = 100 ~ \AA$, $V_{0} = 10^{4} ~ m/s$, 
so that characteristic parameters $t_{0} = 10^{-12} ~ s$, $E_{0} = 1.04 \times 10^{8} ~ V/m$, $T_{0} = 1.21 \times 10^{4} ~ K$ and $\lambda_{D} = 0.76$,
and derived parameters $v_{th} = 0.31 \sim 0.54 ~ V_{0}$, $Kn = 0.48 \sim 4.3$, and $\Gamma = 2.3 \sim 0.77$.
The one-component plasma system is across the weak and moderate coupling regime, where the effects of the particle correlation become non-negligible, and the Landau equation is insufficient to characterize the plasma kinetics.

% For $\Gamma \ll O(1)$, the plasma is in the weakly coupled regime (e.g., high-temperature, low-density), and the small-angle scattering and binary collisions are the dominant interaction mechanisms. The Landau equation provides a valid approximation that remains valid.
% For $\Gamma \sim O(1)$, the plasma is in the moderately coupled regime, and the effects of the particle correlation become non-negligible. 
% The Landau equation is insufficient to characterize the plasma kinetics due to the oversimplification of the collisional interactions in the form of a homogeneous kernel. 
% In this study, we focus on this challenging regime and aim to accurately model the plasma kinetics that accounts for the unresolved particle correlations in the form of a generalized heterogeneous non-stationary collisional kernel \eqref{eq:CM2}. 

In the training process, trajectories with a number of particles $N=10^{6}$, beginning with initial particle velocity distributions at some $(\rho,T)$ pairs, including uniform and the bi-Maxwellian distributions, are used to train the encoder functions in Eq. \eqref{eq:SS} and learn the generalized collision kernel.
The test functions $\psi_{k}(\bm{v}) = \exp[-(\boldsymbol{v}-\mu)^2/2\sigma^{2}]$, $\alpha~\boldsymbol{v}^2 \exp(-\boldsymbol{v}^2/2\sigma^{2})$, $\exp[-(\boldsymbol{v}^2 - \mu^{2})^2/2\sigma^{2}]$ with parameters $\mu$, $\sigma$, are designed to cover the relevant regions of velocity space distribution based on the system temperature.
The detailed information and discussion can be found in reference \cite{zhao2025fast}.

\section{MD settings for the 1D-3V system} 
% Physical characteristic quantity analysis
% Dimensional analysis and characteristic parameters
\label{app:MD}

We consider the 1D-3V plasma kinetics with the initial condition
\begin{equation*}
f(x, \bm v, t=0) = \rho_0  \tilde{f}(\bm{v}; x) 
\end{equation*}
where $\rho_0 = n_0 m_0$ with $n_0 = 10^{24} m^{-3}$ and $\tilde{f}(\bm{v}_{\bm{j}}; x_{i})$ takes various distributions with spatially dependent temperature, i.e.,  
\begin{equation}
\begin{aligned}
    &\int  \tilde{f}(\bm{v}; x) {\rm d} \bm v = 1 , 
    \quad \quad \qquad \quad \int \bm{v}  \tilde{f}(\bm{v}; x) {\rm d} \bm v = \bm 0, \\
    &\int \vert \bm{v} \vert^{2} \tilde{f}(\bm{v}; x) {\rm d} = T(x) , \quad~~
    T(x) = 0.2 + 0.1 \sin(2\pi/L_{x}) ~ \text{eV} ,
\end{aligned}    
\end{equation}
where $\tilde{f}(\bm{v}; x)$ takes the form including the local Maxwellian, bi-Maxwellian, symmetric double-well, and asymmetric double-well distributions as follows
\begin{equation}
\begin{aligned}
    \tilde{f}(\bm{v}; x) \sim& \exp\left( -\dfrac{m_{0} \bm{v}^{2}}{2k_{B}T(x)} \right) , \\
    \tilde{f}(\bm{v}; x) \sim& \exp\left( -\dfrac{m_{0} \bm{v}_{x}^{2}}{2k_{B}(T(x)/3)} - \dfrac{m_{0} \bm{v}_{y,z}^{2}}{2k_{B}(4T(x)/3)} \right) , \\
    \tilde{f}(\bm{v}; x) \sim& \prod_{i=x,y,z} \left[ \exp\left( -\dfrac{(\bm{v}_i-b_{0})^2}{2\sigma_{0}^2} \right) + \exp\left( -\dfrac{(\bm{v}_i+b_{0})^2}{2\sigma_{0}^2} \right) \right], \\
    & b_{0}^{2} = \dfrac{0.8 k_{B} T}{m_{0}}, ~ \sigma_{0}^{2} = \dfrac{0.2 k_{B} T}{m_{0}} , \\
    \tilde{f}(\bm{v}; x) \sim& \prod_{i=x,y,z} \left[ \sigma_{1}^{-1} \exp\left( -\dfrac{(\bm{v}_i-b_{1})^2}{2\sigma_{1}^2} \right) + \sigma_{2}^{-1} \exp\left( -\dfrac{(\bm{v}_i+b_{1})^2}{2\sigma_{2}^2} \right) \right], \\
    & b_{1}^{2} = \dfrac{0.875 k_{B} T}{m_{0}}, ~ \sigma_{1}^{2} = \dfrac{0.05 k_{B} T}{m_{0}}, ~ \sigma_{2}^{2} = \dfrac{0.2 k_{B} T}{m_{0}} .
\end{aligned}
\end{equation}

MD simulations of the 1D-3V system are performed with $N_{MD} = 1.6 \times 10^{7}$ particles confined in a box of side length $1024~\AA \times 125000~\AA \times 125000~\AA$, with the overall periodic boundary condition.
The initial setup of the system comprises an equilibrated configuration and randomly reassigned particle velocities according to the velocity distributions mentioned above.
In the simulation processes, the spatial extent in the x-direction is divided into $20$ grids.

\section*{Acknowledgments}
The work is supported in part by the National Science Foundation under Grant DMS-2110981, the Department of Energy under
Grant DOE-DESC0023164 and the ACCESS program through allocation MTH210005. The
authors also acknowledge the support from the Institute for Cyber-Enabled Research at Michigan
State University.

%\bibliographystyle{plain} % elsart-num, elsart-num-sort
%\bibliography{ref}

%\begin{thebibliography}{00}


%\bibitem[Lamport(1994)]{lamport94}
%  \textit{\LaTeX: a document preparation system},
%  Addison Wesley, Massachusetts,

%\end{thebibliography}

\end{document}